\documentclass[11pt,a4paper]{article}

\usepackage{braket} 
\usepackage{mathtools}
\usepackage{authblk} 
\usepackage{algpseudocode,algorithmicx,algorithm}

\usepackage{mathrsfs}
\usepackage{latexsym,bm}
\usepackage{amsmath,amsfonts,amsmath,amssymb,amsthm}
\usepackage{extarrows}

\usepackage{graphicx,subfigure,epstopdf,float}
\usepackage{enumerate,cases,multirow}
\usepackage{makecell}
\usepackage{caption}
\usepackage{tikz}
\usepackage{pgfplots}

\usepackage{longtable,colortbl,arydshln,threeparttable}
\definecolor{mygray}{gray}{.9}
\usepackage[shortlabels]{enumitem}

\usepackage{indentfirst}
\usepackage[top=25mm,bottom=20mm,left=25mm,right=20mm]{geometry}
\usepackage{cite}

\usepackage{listings}

\usepackage{makeidx}        
\usepackage{booktabs}
\usepackage[bookmarksnumbered,colorlinks,citecolor=red,linkcolor=red,hyperindex,linktocpage=true]{hyperref}

\newtheorem{theorem}{Theorem}[section]
\newtheorem{lemma}{Lemma}[section]

\theoremstyle{remark}
\newtheorem{remark}{\bf Remark}[section]
\theoremstyle{assumption}

\numberwithin{equation}{section}

\begin{document}
\title{Quantum simulation of elastic wave equations via Schr$\ddot{\text{o}}$dingerisation}

\author[1,2,3]{Shi Jin   \thanks{shijin-m@sjtu.edu.cn}}
\author[1]{Chundan Zhang   \footnote{Corresponding author.} \thanks{chundanzhang@sjtu.edu.cn}}
\affil[1]{School of Mathematical Sciences,   Shanghai Jiao Tong University, Shanghai 200240, China.}
\affil[2]{Institute of Natural Sciences, Shanghai Jiao Tong University, Shanghai 200240, China.}
\affil[3]{Ministry of Education Key Laboratory in Scientific and Engineering Computing, Shanghai Jiao Tong University, Shanghai 200240, China.}

\date{\today}
\maketitle

\begin{abstract}
In this paper we study quantum simulation algorithms on the elastic wave equations using the Schr$\ddot{\text{o}}$dingerisation method. The Schr$\ddot{\text{o}}$dingerisation method transforms any linear PDEs into a system of Schr\"odinger-type PDEs --with unitary evolution--using the warped phase transformation that maps the equations in one higher dimension. This makes them  suitable for  quantum simulations.  We expore the application in two forms of the elastic wave equations. For the velocity-stress equation in isotropic media, we explore the symmetric matrix form under the external forcing via Schr$\ddot{\text{o}}$dingerisation combined with spectral method. For problems with variable medium parameters, we apply Schr$\ddot{\text{o}}$dingerisation method based on the staggered grid method to simulate velocity and stress fields, and give the complexity estimates. For the wave displacement equation, we transform it into a hyperbolic system and apply the Schr$\ddot{\text{o}}$dingerisation method, which is then discretized by the spectral method and central difference scheme. Details of the quantum algorithms will be provided, along with  the complexity analysis which demontrate exponential quantum advantage in space dimensin over the classical algorithms. 
\end{abstract}

\textbf{Keywords}: 
	Elastic wave equations,  quantum algorithm, Schr$\ddot{\text{o}}$dingerisation method
	
\section{Introduction}

The elastic wave equations describe the propagation  of strain and stress generated inside an elastic object, subject to external disturbance or external force. In the 1820s, Navier proposed a mathematical model to describe the equilibrium and motion of elastic bodies, which opened up the study of elastic waves\cite{navier1827}. Madariaga\cite{madariaga} first modeled velocity-stress formulation in 1976. Elastic wave theory reveals the law of force propagation in the medium, which is of great importance for understanding various types of fluctuations such as volcanism and geodynamics\cite{Quantitativeseismology,bullen1985introduction}. 

Classical methods for solving elastic wave equations include finite difference methods \cite{leveque2007finite}, finite element methods \cite{finite-element}, boundary element methods \cite{boundaryelement} and other numerical methods\cite{FGAChai2018}. However, when one solves such equations in a large domain with the requirement of high accurate or resolution, the classical computation becomes exhaustingly expensive. With the rapid development of quantum computers and chips, quantum algorithms for scientific computing have become increasingly promising\cite{Ber14,Childs20,Cao13}, providing a promising alternative for the large scale simulation of elastic wave equations. 

There are several quantum algorithms designed for solving elastic wave or related equations in the literature. For examples,  
   \cite{PCS30,suau31} employs Hamiltonian simulation and quantum linear system algorithms (QLSA) to solve wave equations in homogenous medium. \cite{sato33} uses the linear combination of Hamiltonian simulation (LCHS)\cite{LCHS} techniques to  solve second-order linear PDEs of nonconservative systems with spatially varying parameters. \cite{matle2} presents a quantum computing framework for the velocity-stress equations with source and loss functions. \cite{quan_cauchy} solves the Cauchy problem for symmetric first order linear hyperbolic systems by digital quantum algorithm and reservoir technique. 

A recently developed  method for quantum simulation of general linear PDEs is Schr$\ddot{\text{o}}$dingerisation\cite{Schrodingerization}. This method maps all linear partial differential equations to one  higher dimension using the wraped transformation, while the new equations become a "Schr$\ddot{\text{o}}$dingerised" or Hamiltonian system that evolves in unitary dynamics, which can be solved on a quantum computer using Hamiltonian simulation. There have been well-established study of Schr$\ddot{\text{o}}$dingerisation method in the complex situation of practical physical problems, for example, boundary conditions\cite{boundarycondition_interface_conditions,Artificial_Boundary_Conditions}, interface problems\cite{boundarycondition_interface_conditions} and PDEs with source terms\cite{inhomogeneous_schr}. It has also been  applied to Maxwell’s equations\cite{Maxwell}, the Fokker-Planck equations\cite{Fokker_Planck}, highly-oscillatory transport equations\cite{oscillatory_transport} and ill-posed problems such as the backward heat equation\cite{ill-posed}. Its versatility extends to multiscale systems\cite{Multiscale}, non-autonomous differential systems\cite{time-dependentHamiltonians}, iterative linear algebra solvers\cite{iterative} and stochastic differential equations\cite{sde}.

In this work, we investigate the applicability of the Schr$\ddot{\text{o}}$dingerisation method across two fundamental models: the symmetric matrix form (SMF)\cite{SMFd3} of velocity-stress equation, and the generalized variable-coefficient counterparts-- the hyperbolic system of displacement equation\cite{ryzhik1996transport}. We only consider periodic boundary conditions and leave other boundary value problems to the future.  The main contributions are the following:

\begin{enumerate}
    \item [(a)] We resolve non-unitary dynamics induced by external forces in symmetric matrix equations through Schr$\ddot{\text{o}}$dingerisation, employing the spectral method for spatial approximation. Furthermore, the Schr$\ddot{\text{o}}$dingerisation method, combined with the staggered grid method, proves particularly effective for variable-coefficient cases, offering exponential advantage in high-dimensional space.  We also  provide a theorem on the gate complexity of the algorithm. We note that the staggered grid method\cite{1986P,Randall} is the most commonly used difference method to solve the velocity-stress equation. 
    \item [(b)] We apply the Schr$\ddot{\text{o}}$dingerisation method to the asymmetry hyperbolic system combined with the spectral method and the central difference method, giving  a theorem showing quantum exponential advantage. For spectral discretization and central difference discretization, we conclude a practically efficient criterion for selecting the region of extended space $p$, which facilitates numerical simulations.

 \end{enumerate} 

 Several numerical experiments will be conducted on classical computers to verify the correctness of these quantum algorithms.

The structure of the paper is as follows: In Section \ref{sec:reviewEWE}, we review the derivation of the elastic wave equations. 
Section \ref{sec:review_Schr} provides an overview of the Schr$\ddot{\text{o}}$dingerisation method, encompassing its discrete and continuous frameworks, the recovery of solutions, and the complexity results. Section \ref{sec:vseq} focuses on the quantum simulation of the SMF equations and variable parameter equations. In Section \ref{sec:QS1}, we apply the Schr$\ddot{\text{o}}$dingerisation method to the hyperbolic system of displacement equations.  Finally, Section \ref{sec:Numricalsimulation} presents numerical examples to verify the correctness of the Schr\"odingerization methods.

In this paper, we use $M$, $N$ and $N_{\xi}$ as the grid numbers of $x-$ (space), $p-$ (extended variable in the Schr\"odingeriation method) and $\xi-$ (the Fourier variable) directions respectively, where $M = 2^m$, $N = 2^n$ and $N_{\xi} = 2^{n_{\xi}}$ are the even integers. We write the sequence $j = {0, 1, 2,\cdots, M-1}$ as $j\in [M]$ and represent the all-zero matrix and identity matrix by the boldfaced $\boldsymbol{0}$ and $\boldsymbol{1}$. The matrix dimensions can be inferred from context.

\section{A review of elastic wave equations}\label{sec:reviewEWE}

The motion of a small displacement $\boldsymbol{u}(t,\boldsymbol{x})$ in  an elastic medium can be described as:
\begin{equation}\label{disp_stress}
    \rho(\boldsymbol{x}) \frac{d^2 \boldsymbol{u}}{d t^2} = L
    \boldsymbol{\sigma} + \boldsymbol{f},
\end{equation}
where $\boldsymbol{u} = (u_1, u_2, u_3)^\top$ is the displacement vector, $ \boldsymbol{\sigma}=(\sigma_{11}, \sigma_{22}, \sigma_{33}, \sigma_{12}, \sigma_{13}, \sigma_{23})^\top $ is the stress vector and $ \boldsymbol{f} = (f_1,f_2,f_3)^\top $ is the body force
in  three dimensions. $L$ is the partial derivative operator,
$$
L=\left(\begin{array}{cccccc}
k_1 & 0 & 0 & k_2 & k_3 & 0 \\
0 & k_2 & 0 & k_1 & 0 & k_3 \\
0 & 0 & k_3 & 0 & k_1 & k_2
\end{array}\right)
$$
where the subscript $i$ denotes the different spatial direction and $k_i = \frac{\partial}{\partial {x_i}}$. In addition, the particle motion produces displacement and deformation, and the relationship between displacement and strain is called the Cauchy equation: 
\begin{equation}\label{disp_strain}
    \boldsymbol{\epsilon}=L^\top \boldsymbol{u}, 
\end{equation}
where $\boldsymbol{\varepsilon}=(\varepsilon_{11}, \varepsilon_{22}, \varepsilon_{33}, \varepsilon_{12}, \varepsilon_{13}, \varepsilon_{23})$ is the strain vector. Hooke's Law defines the stress-strain relationship as a constitutive equation: 
\begin{equation}\label{stress_strain}
    \boldsymbol{\sigma}  = \mathbf{C} \boldsymbol{\epsilon},
\end{equation} 
where $\mathbf{C}$ is the stiffness matrix, which consists of the elastic parameters. In an isotropic medium,  
$$
\mathbf{C}=\begin{pmatrix}
\lambda+2\mu & \lambda & \lambda & & & \\
\lambda& \lambda+2\mu & \lambda & & & \\
\lambda& \lambda& \lambda+2\mu & & & \\
& & & \mu & & \\
& & & & \mu & \\
& & & & &\mu
\end{pmatrix},
$$
where $\lambda$ and $\mu$ are called Lam$\acute{e}$ parameters.
By substituting the Cauchy equation \eqref{disp_strain} and constitutive equation \eqref{stress_strain} into the motion equation \eqref{disp_stress}, the displacement expression of the wave equation can be obtained:
\begin{equation}\label{eq4}
\rho \frac{\partial^2 \mathbf{u}}{\partial t^2}=L C L^T \mathbf{u}+\rho \mathbf{F}.
\end{equation}

Another first-order system describing the elastic wave equation is the velocity-stress equation. By introducing the velocity variable 
\begin{equation}\label{v}
    \boldsymbol{v} = \frac{\partial \boldsymbol{u}}{\partial t},
\end{equation}
the elastic wave equation can be written in the form of the velocity-stress equation
\begin{subequations}\label{velocity-stress}
\begin{align}
     \rho(\boldsymbol{x}) \frac{d \boldsymbol{v}}{d t} &= L
    \boldsymbol{\sigma} + \boldsymbol{F},\\
    \boldsymbol{\sigma}_t  &= \mathbf{C} L^\top \boldsymbol{v}. 
\end{align}
\end{subequations}

The velocity-stress equation can provide both vibration velocity information and stress information. In particular, it takes into account the spatial variation of the Lam$\acute{e}$ parameters and does not directly take derivatives, which is more suitable for the calculation of elastic wave equation in heterogeneous media.
\subsection{The first-order form of the elastic wave equation}
In this paper, we consider the elastic wave equation in isotropic media. In the case of homogeneous media, we use the SMF of the velocity-stress equation, which is suitable for the wave propagation in the curve domain and significantly simplifies the complexity of the formulation and programming\cite{SMFd2,SMFd3}. For the wave displacement equation, we consider the hyperbolic systems typically used for high-frequency analysis of elastic waves\cite{jin2006hamiltonian,ryzhik1996transport}, which, despite asymmetry, reveal quantum advantages through Schr$\ddot{\text{o}}$dingerisation.

\subsubsection{The Symmetric Matrix Form Framework for the
Elastic Wave Equation}\label{sec:2.1.2}
If the medium is homogeneous, the velocity-stress equation has the symmetric matrix form. By separating the different spatial partial derivative operator $L$ into $L_i$,  equation \eqref{velocity-stress} can be written as
\begin{subequations}\label{v-s_withoutforce}
    \begin{align}
        \rho \frac{d \boldsymbol{v}}{dt} &= \sum_{i=1,2,3}  k_i L_i \boldsymbol{\sigma}+ \boldsymbol{f},\\
        \frac{d \boldsymbol{\sigma}}{dt} &= \boldsymbol{C}\sum_{i=1,2,3}  k_i L_i^T \boldsymbol{v},
    \end{align}
\end{subequations}
where $\boldsymbol{C}$ is the stiffness matrix and the density $\rho$ is invertible and independent of $\boldsymbol{x}$. $L_1, L_2$ and $L_3$ are as follows:
\begin{equation}
L_1 = \left[\begin{array}{cccccc}
1 & 0 & 0 & 0 & 0 & 0 \\
0 & 0 & 0 & 1 & 0 & 0 \\
0 & 0 & 0 & 0 & 1 & 0
\end{array}\right], \quad
L_2 = \left[\begin{array}{cccccc}
0 & 0 & 0 & 1 & 0 & 0 \\
0 & 1 & 0 & 0 & 0 & 0 \\
0 & 0 & 0 & 0 & 0 & 1
\end{array}\right], \quad 
L_3=\left[\begin{array}{cccccc}
0 & 0 & 0 & 0 & 1 & 0 \\
0 & 0 & 0 & 0& 0 & 1 \\
0 & 0 & 1 & 0 & 0 & 0
\end{array}\right].
\end{equation}
Let $\tilde{U} = (\boldsymbol{\sigma}/\rho, \boldsymbol{v}) = (\frac{\sigma_{11}}{\rho}, \frac{\sigma_{22}}{\rho}, \frac{\sigma_{33}}{\rho}, \frac{\sigma_{12}}{\rho}, \frac{\sigma_{13}}{\rho}, \frac{\sigma_{23}}{\rho}, v_1, v_2, v_3)^T$ and  $\tilde{\mathbf{f}} =(\boldsymbol{0}_{6\times1},\boldsymbol{f}/\rho) $, then Eqs. \eqref{v-s_withoutforce} can be written as
\begin{equation}\label{v-s_titlde}
    \widetilde{\mathbf{A}}_0 \frac{\partial \widetilde{\mathbf{U}}}{\partial t}=\widetilde{\mathbf{A}}_x \frac{\partial \widetilde{\mathbf{U}}}{\partial x}+\widetilde{\mathbf{A}}_y \frac{\partial \widetilde{\mathbf{U}}}{\partial y}+\widetilde{\mathbf{A}}_z \frac{\partial \widetilde{\mathbf{U}}}{\partial z} + \widetilde{\mathbf{f}},
\end{equation}
where 
\begin{equation}
\widetilde{\mathbf{A}}_x = \begin{pmatrix}
        &L_1^T\\
        L_1&
    \end{pmatrix}, 
\widetilde{\mathbf{A}}_y = \begin{pmatrix}
        &L_2^T\\
       L_2&
    \end{pmatrix},
\widetilde{\mathbf{A}}_z = \begin{pmatrix}
        &L_3^T\\
        L_3&
    \end{pmatrix},
\end{equation}
\begin{equation}
\widetilde{\mathbf{A}}_0 =\begin{bmatrix}
        \rho\boldsymbol{C}^{-1}& \\
         &I_3
    \end{bmatrix} ,\quad
    \rho\boldsymbol{C}^{-1}=\left[\begin{array}{llllll}
    a & b & b & 0 & 0 & 0  \\ 
    b & a & b & 0 & 0 & 0 \\ 
    b & b & a & 0 & 0 & 0 \\
    0 & 0 & 0 & c & 0 & 0  \\ 
    0 & 0 & 0 & 0 & c & 0 \\ 
    0 & 0 & 0 & 0 & 0 & c  
    \end{array}\right].
\end{equation}
Here $a=\frac{c_p^2-c_s^2}{c_s^2\left(3 c_p^2-4 c_s^2\right)}, b=-\frac{c_p^2-2 c_s^2}{c_s^2\left(6 c_p^2-8 c_s^2\right)}, c=\frac{1}{c_s^2}$ are denoted by the longitudinal and transverse medium wave speed,
$$c_p = \sqrt{\frac{\lambda+2\mu}{\rho}}, c_s = \sqrt{\frac{\mu}{\rho}}.$$
As $\widetilde{\mathbf{A}}_0$ is positive definite, ansatz $\widetilde{\mathbf{A}}_0 = P^T \Lambda P$, $P$ is an orthonormal matrix. Thus, we have
$$ \boldsymbol{\Lambda}=\left[\begin{array}{cccccc}\frac{1}{3 c_p^2-4 c_s^2} & & & & & \\ & \frac{1}{2 c_s^2} & & & & \\ & & \frac{1}{2 c_s^2} & & & \\ & & & 1 / c_s^2 & & \\ & & & & 1 / c_s^2 & \\ & & & & & 1 / c_s^2\end{array}\right], \mathbf{P}=\left[\begin{array}{cccccc}\frac{1}{\sqrt{3}} & \frac{1}{\sqrt{3}} & \frac{1}{\sqrt{3}} & 0 & 0 & 0 \\ \frac{1}{\sqrt{2}} & -\frac{1}{\sqrt{2}} & 0 & 0 & 0 & 0 \\ \frac{1}{\sqrt{6}} & \frac{1}{\sqrt{6}} & -\frac{2}{\sqrt{3}} & 0 & 0 & 0 \\ 0 & 0 & 0 & 1 & 0 & 0 \\ 0 & 0 & 0 & 0 & 1 & 0 \\ 0 & 0 & 0 & 0 & 0 & 1\end{array}\right].$$
Denote $M = P^T \Lambda^{-\frac{1}{2}} P$, and pre-multiply $M$ from Eqs. 
 \eqref{v-s_titlde} yields 
\begin{equation}\label{v-s_SMF}
    \frac{\partial \mathbf{U}}{\partial t}=\mathbf{A}_x \frac{\partial \mathbf{U}}{\partial x}+\mathbf{A}_y \frac{\partial \mathbf{U}}{\partial y}+\mathbf{A}_z \frac{\partial \mathbf{U}}{\partial z}+\mathbf{f},
\end{equation}
where $\mathbf{U} = M^{-1}\widetilde{\mathbf{U}}$, $\mathbf{f} =  M^{-1}\widetilde{\mathbf{A}}_0^{-1}\widetilde{\mathbf{f}}$ and $\mathbf{A}_x = M\widetilde{\mathbf{A}}_xM$ is symmetric, as are $\mathbf{A}_y$ and $\mathbf{A}_z$. Eq \eqref{v-s_SMF} is the symmetric matrix form of the elastic wave equations.

\subsubsection{The hyperbolic system of wave displacement equation}

Eqs. \eqref{disp_stress}-\eqref{stress_strain} can be written as follows:
\begin{subequations}\label{disp_stress_iso}
\begin{align}
&\rho(x) \frac{d^2 u_i}{d t^2} = \frac{\partial \sigma_{ij}}{\partial x_j} + f_i,\\
&\sigma_{i j} = \lambda(\boldsymbol{x}) \frac{\partial u_k}{\partial x_k} \delta_{i j} + \mu(\boldsymbol{x})\left(\frac{\partial u_i}{\partial x_j} + \frac{\partial u_j}{\partial x_i}\right),
\end{align}
\end{subequations}
where $\delta_{i j}$ is the  Kronecker delta ($=1$ if $i = j$; $=0$ if $i \neq j$). It can be written as a hyperbolic system with respect to the first derivative of time \cite{ryzhik1996transport}. By Introducing three new variables: 
$$
p=\lambda \operatorname{div} \mathbf{u}, \quad
\xi_i=\dot{u}_i, \quad 
\zeta_{i j}=\mu\left(\frac{\partial u_i}{\partial x_j}+\frac{\partial u_j}{\partial x_i}\right), \quad
i, j=1,2,3.
$$
Eq. \eqref{disp_stress_iso} then becomes:
\begin{equation}\label{Tran1}
\rho \dot{\xi}_i  =\frac{\partial p}{\partial x_i}+\sum_j \frac{\partial \zeta_{i j}}{\partial x_j} + f_i, \quad
\dot{\zeta}_{i j}  =\mu\left(\frac{\partial \xi_i}{\partial x_j}+\frac{\partial \xi_j}{\partial x_i}\right), \quad
\dot{p}  =\lambda \operatorname{div} \boldsymbol{\xi} .
\end{equation}
Let $\boldsymbol{\xi} = (\xi_1,\xi_2,\xi_3)$, $\boldsymbol{\zeta_1} = (\zeta_{11},\zeta_{22},\zeta_{33})$, $\boldsymbol{\zeta_2} = (\zeta_{23},\zeta_{13},\zeta_{12})$, $\boldsymbol{f}=(f_1, f_2, f_3)$. Then eq. \eqref{Tran1} can be written as
\begin{equation}\label{secondorder_EWE}
    \begin{aligned}
        \frac{\partial}{\partial t} \boldsymbol{\xi}
			- \frac{1}{\rho}K(\boldsymbol{k})\boldsymbol{\zeta_1} - \frac{1}{\rho}M(\boldsymbol{k}) \boldsymbol{\zeta_2} - \frac{1}{\rho}\boldsymbol{k}p
			&= - \boldsymbol{f}, \quad
   \frac{\partial}{\partial t} p
			-\lambda \boldsymbol{k}^T\boldsymbol{\xi} =0,\\ 
			\frac{\partial}{\partial t}\boldsymbol{\zeta_1}
			-2\mu K(\boldsymbol{k})\boldsymbol{\xi} &=0,\quad
   \frac{\partial}{\partial t}\boldsymbol{\zeta_2}
			-\mu M(\boldsymbol{k})\boldsymbol{\xi} =0,
    \end{aligned}
\end{equation}
where the matrix $K(\boldsymbol{k})=\operatorname{diag}(\boldsymbol{k}) = \operatorname{diag}\left(k_1, k_2, k_3\right)$ and
\begin{equation}
    M(\boldsymbol{k})=\left(\begin{array}{ccc}
0 & k_3 & k_2 \\
k_3 & 0 & k_1 \\
k_2 & k_1 & 0
\end{array}\right).
\end{equation}
By defining $\mathbf{w} = (\boldsymbol{\xi},\boldsymbol{\zeta_1},\boldsymbol{\zeta_2},p)$, $\boldsymbol{F} = (\boldsymbol{f},\boldsymbol{0}_{7\times 1})$, the hyperbolic matrix version of Eq. \eqref{secondorder_EWE}, which can be used to study the high frequency approximation for elastic wave, then becomes \cite{jin2006hamiltonian}, 
\begin{equation}\label{eq5}
\frac{\partial \mathbf{w}}{\partial t}+L(\mathbf{x}, \mathbf{k})\mathbf{w} + \boldsymbol{F} =0.
\end{equation}
The dispersion matrix $L(\mathbf{x}, \mathbf{k})$ is defined by
\begin{equation}
    L=-\left(\begin{array}{cccc}
0 & K(\boldsymbol{k}) / \rho & M(\boldsymbol{k}) / \rho & \frac{1}{\rho} \boldsymbol{k} \\
2 \mu K(\boldsymbol{k}) & 0 & 0 & 0 \\
\mu M(\boldsymbol{k}) & 0 & 0 & 0 \\
\lambda \boldsymbol{k}^t & 0 & 0 & 0
\end{array}\right).
\end{equation}

We point out that this formulation is valid  only for initial value problems or period boundary conditions. In general,  $u_x |_{x=0/L}$ can not be obtained  from original boundary conditions for other physicalk boundary conditions. This remains an open issue. 

\section{A review of the Schr$\ddot{\text{o}}$dingerisation methed}\label{sec:review_Schr}

In this section, we will briefly review the quantum Schr$\ddot{\text{o}}$dingerisation method for solving linear partial differential equations, including both the  discrete and continous Schr$\ddot{\text{o}}$dingerisation methods, and give two recovery options after obtaining the solution. Finally, the complexity analysis is conducted.

For any linear partial differential equations, it can be transformed into a linear partial differential equation with respect to the first order of time by introducing new variables or other transformations, and then ordinary differential equations can be obtained after spatial discretizations: 
\begin{equation}\label{ode}
    \frac{d \boldsymbol{u}}{dt} = A\boldsymbol{u} + \boldsymbol{b},\quad \boldsymbol{u}(0) = \boldsymbol{u}_0, 
\end{equation}
where $\boldsymbol{u} = \boldsymbol{u}(t)\in \mathbb{C}^n$, $A\in \mathbb{C}^n\times \mathbb{C}^n$ is obtained after spatial discretization and can be time-dependent. $\boldsymbol{b}\in \mathbb{C}^n$ may arise from boundary conditions or external source terms. In order to simulate \eqref{ode} on a quantum computer, one first puts it into a  homogeneous form: 
\begin{equation}\
    \frac{d}{dt}\begin{pmatrix}
        \boldsymbol{u}\\
        \boldsymbol{r}
    \end{pmatrix} = \begin{pmatrix}
        A&\operatorname{diag}(\boldsymbol{b})/\epsilon\\
        \boldsymbol{0}&\boldsymbol{0}
    \end{pmatrix}\begin{pmatrix}
        \boldsymbol{u}\\
        \boldsymbol{r}
    \end{pmatrix},
\end{equation}
where $\boldsymbol{r}= \epsilon \boldsymbol{1}$ is a constant vector and $\epsilon \sim \Vert \boldsymbol{b} \Vert_\infty$ which can eliminate the influence of inhomogeneous terms on the eigenvalues of the new coefficient matrix\cite{inhomogeneous_schr}.

Without loss of generality, let's consider Eq. \eqref{ode} with $\boldsymbol{b}=0$ and use the  decomposition:
\begin{equation}\label{ode_H}
    \frac{d}{dt}\boldsymbol{u} = A\boldsymbol{u} = (H_1 + iH_2)\boldsymbol{u}, 
\end{equation}
where the Hermitians  $H_1 = \frac{A+A^\dagger}{2}$, and $H_2 = \frac{A-A^\dagger}{2i}$. Considering the imaginary unit $i$ in Schr$\ddot{\text{o}}$dinger equation, one  can introduce a new variable $p$ and use the warped phase transformation. Let $v = e^{-|p|}\boldsymbol{u}$, Eq. \eqref{ode_H} is equivalent to 
\begin{equation}\label{p_v}
    \frac{dv}{dt} = -H_1\partial_p v +i H_2 v,\quad v(0) = e^{-|p|}\boldsymbol{u}(0). 
\end{equation}
In \cite{ill-posed}, a smooth initial data $v(0) = g(p)\boldsymbol{u}(0)$ is used to improve the algorithm accuracy,
\begin{equation}\label{g(p)}
g(p)= \begin{cases}h(p) & p \in(-\infty, 0] \\ e^{-p} & p \in(0, \infty)\end{cases}\\
\end{equation}
with $h(p) \in L^2((-\infty, 0))$.

\subsection{The Schr$\ddot{\text{o}}$dingerisation method}

Suppose that the $p$ domain is $[-\pi L,\pi L]$, partitioned by mesh points  $-\pi L = p_0 < p_1 <...< p_N = \pi L$. The Fourier basis functions are: 
\begin{equation}
    \phi_l(p) = e^{i\mu_l(p+\pi l)}, \quad \mu_l = (l-N/2)/L, \quad l\in[N].
\end{equation}
Taking the discrete Fourier transform on Eq. \eqref{p_v} in the $p$ domain gives
\begin{equation}\label{DFT_P}
\begin{cases}
    \frac{d }{dt}\boldsymbol{v}_h&= -i(H_1\otimes P_{\mu})\boldsymbol{v}_h+i(H_2\otimes I)\boldsymbol{v}_h,\\
    \hat{\boldsymbol{v}}_h(0) &= \boldsymbol{u}(0)\otimes \sum_{j=0}^N e^{-|p_j|}\ket{j}.
\end{cases}
\end{equation}
In order to reduce the complexity in quantum simulations, one can use the inverse quantum Fourier transform to change the variables, thus reduce the sparsity of the Hamiltonian. Let $\boldsymbol{c} = I\otimes \Phi^{-1} \boldsymbol{v}_h$, one gets
\begin{equation}\label{DFT_D}
    \begin{cases}
    \frac{d }{dt}\boldsymbol{c}&= -i(H_1\otimes D_{\mu})\boldsymbol{c}+i(H_2\otimes I)\boldsymbol{c},\\
    \boldsymbol{c}(0) &= I\otimes \Phi^{-1} \boldsymbol{v}_h(0).
\end{cases}
\end{equation}
This Schr\"odingeirzed system is suitable for quantum simulation  with  Hamiltonian $H=-i(H_1\otimes D_{\mu}-H_2\otimes I)$. Here $P_{\mu}$ and $D_{\mu}$ are the Fourier correlation matrix of the variable $p$, and the details of the Fourier spectral method can be found in Sec \ref{sec:SMFofvseq} of this paper.
\subsection{Recovery of solutions}
One can obtain quantum state $\ket{\boldsymbol{v}_h(T)}$ by applying the Quantum Fourier Transform (QFT)  on the solution $\ket{\boldsymbol{c}(T)}$ of Eq. \eqref{DFT_D},
\begin{equation}\label{recovery_formula}
     \boldsymbol{v}_h(T) = (I\otimes \Phi) \boldsymbol{c}(T), 
\end{equation}
where $\boldsymbol{v}_h$ is the lattice value of $v$ after discretizing $x$ and $p$. Depending on the definition of $v$, we can choose either multiplying by $e^p$  at one point of $p$, or integrating to recover $\boldsymbol{u} $:
\begin{equation}\label{p^*}
    \boldsymbol{u} = e^{p^*} v, \quad \text{for any} \quad  p\geq p^*; \quad \text{or} \quad \boldsymbol{u} = \int_{p^*}^{\infty} v(p) dt
\end{equation}
where $p^* \geq \max \left\{\lambda_n\left(H_1\right) T, 0\right\}$. 

\subsection{The complexity}
\begin{lemma}\label{lem:ham_simu complexity}
\cite{OptimalHS} For evolution time $T$ and error $\epsilon$, a time-independent $s$-sparse Hamiltonian $H$ on $m_H$ qubits subject to $m_e$ bits of precision, has the query
complexity with respect to $\tau = s\|H\|_{max} T$ as  
\begin{equation}
    N_{query} = \mathscr{O}\bigg(\tau+ \log(1/\delta)/(\log\log(1/\delta))\bigg),
\end{equation}
which achieves known lower bounds additively with an additional
primitive gate complexity
\begin{equation}
    N_{gate} = \mathscr{O}((m_H + m_e \operatorname{polylog}(m_e))N_{query}).
\end{equation}
\end{lemma}
\begin{remark}
Another Schr\"odingerisation version is using the continuous Fourier transform to Eq. \eqref{p_v}:
    \begin{equation}\label{CFT_xi}
\begin{cases}
    \frac{d }{dt}\hat{v}&= -i\xi H_1\hat{v}+iH_2\hat{v},\\
    \hat{v}(0) &= \frac{1}{\pi(1+\xi^2)}\boldsymbol{u}(0),
\end{cases}
\end{equation}
where $\xi \in \mathbb{C}$ is the variable in the frequency domain corresponding to $p$. Eqs. \eqref{CFT_xi} can be simulated on analog quantum devices in the qumode representation\cite{analog}. If truncating and discretizing $\xi$, $-X= \xi_0\leq \xi_1 \cdots \leq \xi_{N_{\xi}} = X$, one can obtain a Hamiltonian system that is equivalent Eqs. \eqref{DFT_D}:
\begin{equation}\label{CFT_D}
\begin{cases}
    \frac{d }{dt}\hat{\boldsymbol{v}}_h&= -i(H_1\otimes D_{\xi})\hat{\boldsymbol{v}}_h+i(H_2\otimes I)\hat{\boldsymbol{v}}_h,\\
    \hat{\boldsymbol{v}}_h(0) &= \boldsymbol{u}(0) \otimes \sum_{j=0}^{N_{\xi}} \frac{\ket{j}}{\pi(1+\xi_j^2)}.
\end{cases}
\end{equation}

In the continuous version, the recovery of the solution requires the inverse Fourier transform first. Although there will be truncation errors and numerical integration errors during the process, it can be proved that the complexity of the two is still the same. All subsequent derivations in this paper are based on the discrete version.
\end{remark}
\section{Quantum simulation of the first-order velosity-stress equation}\label{sec:vseq}

In this section, we will use the Schr$\ddot{\text{o}}$dingerisation method for quantum simulation of elastic waves using the velocity-stress equations. The first section focuses on the symmetric matrix form, and the second section considers variable coefficient problem.

\subsection{Quantum simulation of SMF of elastic wave equation \eqref{v-s_SMF} by spectral method}\label{sec:SMFofvseq}

We consider the symmetric matrix form of elastic wave equations with time-independent source terms. We can treat the spatial direction similarly to the $p$ direction, using the Fourier spectral discretization. 

Consider the three dimensional domain $\boldsymbol{x} \in \Omega = [a,b]^3$. Denote the discrete grid points as $\boldsymbol{x_j} = (x_{j_1}, y_{j_2},z_{j_3})$, where 
\begin{equation}
    x_{j_1} =a+j_1\Delta x, y_{j_2}=a+j_2\Delta y, z_{j_3}=a+j_3\Delta z,\quad j_1,j_2,j_3 \in [ M ].
\end{equation}
The discrete solution and the discrete source term at a given spatial grid point are:
\begin{equation}
    \boldsymbol{\mathbf{U}}_{i,\boldsymbol{j}} = \mathbf{U}_i(t,\boldsymbol{x_j}), \quad\mathbf{f}_{i,\boldsymbol{j}} = \mathbf{f}_i(\boldsymbol{x_j}),\quad \boldsymbol{j} = \ket{j_1} \otimes \ket{j_2} \otimes \ket{j_3}, \quad i\in [9].
\end{equation}
We use $\boldsymbol{\mathbf{U}}_h$ and $\mathbf{f}_h$ to collect function values:
\begin{equation}
    \boldsymbol{\mathbf{U}}_h = \sum_{i=0}^{8}\sum_{j_1,j_2,j_3 = 0}^{M-1} \mathbf{U}_i(t,\boldsymbol{x_j})\ket{i},\quad
\mathbf{f}_h = \sum_{i=0}^{8}\sum_{j_1,j_2,j_3 = 0}^m\mathbf{f}_i(\boldsymbol{x_j})\ket{i}.
\end{equation}
The one-dimensional basis functions of the Fourier spectral methods are  given by
$$
\phi_l(x)=\mathrm{e}^{\mathrm{i} \mu_l(x-a)}, \quad \mu_l=\frac{2 \pi l}{b-a}, \quad l=-M/2, -M/2+1, \cdots, M/2-1.
$$
For convenience, we adjust the index as:
\begin{equation}
    \phi_l(x)=\mathrm{e}^{\mathrm{i} \mu_l(x-a)}, \quad \mu_l=\frac{2 \pi(l-M/2-1)}{b-a}, \quad 1 \leq l \leq M.
\end{equation}
The matrix involved is
\begin{equation}\label{fourier_spectral_D}
    \Phi=\left(e^{i \mu_l x_j}\right)_{M \times M}, \quad D_\mu=\operatorname{diag}\left\{\mu_1, \cdots, \mu_M\right\}.
\end{equation}
The momentum operator in quantum mechanics can be discretized as a matrix:
\begin{equation}
    \hat{p}q(x) = -i\partial_x q(x),\quad \hat{p}^d\boldsymbol{q} = \Phi D_\mu \Phi^{-1} \boldsymbol{q} =: P_{\mu}\boldsymbol{q}.
\end{equation}
For $d$ dimensions, the basis functions are written as $\phi_{\boldsymbol{l}}(\boldsymbol{x})=\phi_{l_1}(x_1)\cdots\phi_{l_d}(x_d)$, $\boldsymbol{l} = (l_1,\cdots, l_d)$, and the momentum operator is
\begin{equation}
    \hat{P}_lq(\boldsymbol{x}) = -i\partial_{x_l} q(\boldsymbol{x}),\quad \hat{P}_l^d\boldsymbol{q} = \left(I^{\otimes^{l-1}} \otimes P_\mu \otimes I^{\otimes^{d-l}}\right) \boldsymbol{q} =: \boldsymbol{P}_l\boldsymbol{q}.
\end{equation}
Note that
\begin{equation}
    \left(\Phi^{\otimes^d}\right)^{-1} \boldsymbol{P}_l \Phi^{\otimes^d}=I^{\otimes^{l-1}} \otimes D_\mu \otimes I^{\otimes^{d-l}}=: \boldsymbol{D}_l .
\end{equation}
Taking the discrete Fourier transform of the spatial variables of \eqref{v-s_SMF},
\begin{equation}\label{SMF_spectral_P}
    \frac{d}{dt} \boldsymbol{\mathbf{U}}_h = i(\mathbf{A}_x\otimes \boldsymbol{P}_1+\mathbf{A}_y\otimes \boldsymbol{P}_2+\mathbf{A}_z\otimes \boldsymbol{P}_3) \boldsymbol{\mathbf{U}} + \mathbf{f}_h,
\end{equation}
where $\boldsymbol{P}_i = I^{\otimes^{i-1}}\otimes P_{\mu} \otimes I^{\otimes^{3-i}}$. One can improve the efficiency of quantum simulations by reducing the sparsity of $P_{\mu}$, $ P_{\mu} = \Phi D_{\mu} \Phi^{-1}$. Here $P_{\mu}$ and $D_{\mu}$ are the matrices in the discrete $x$-variables. After the substitution of variables, Eq. \eqref{SMF_spectral_P} becomes
\begin{equation}\label{SMF_spectral_D}
     \frac{d}{dt} \hat{\boldsymbol{\mathbf{U}}}_h = i(\mathbf{A}_x\otimes \boldsymbol{D}_1+\mathbf{B}_x\otimes \boldsymbol{D}_2+\mathbf{C}_x\otimes \boldsymbol{D}_3) \hat{\boldsymbol{\mathbf{U}}}_h+{\color{red}\hat{\mathbf{f}}_h}\triangleq A\hat{\boldsymbol{\mathbf{U}}}_h+\hat{\mathbf{f}}_h,
\end{equation}
where $\hat{\mathbf{f}}_h = (I\otimes \Phi)\mathbf{f}_h$.

\begin{remark}
    If $\mathbf{f}_h=\mathbf{0}$, the equation \eqref{SMF_spectral_D} is a Hamiltonian system, which, with the assumptions of homogeneous media, 
  in absence of external sources, and periodic propagation (as used in our derivation), can be directly simulated using Hamiltonian Simulation. However, as we demonstrate below, the Schr$\ddot{\text{o}}$dingerization approach remains essential for solving the systems with source terms and variable coefficients, where traditional Hamiltonian simulations may not suffice.

\end{remark}
\subsubsection{Schr$\ddot{\text{o}}$dingerisation Simulation}\label{sec:vs1SchrSim}
Now, the above ODEs can be rewritten as a
homogeneous system:
\begin{equation}\label{SMF_spectral_ode}
\frac{d}{dt} \begin{pmatrix}
    \hat{\boldsymbol{\mathbf{U}}}_h\\
    \boldsymbol{r}
\end{pmatrix} = \begin{pmatrix}
    A&\operatorname{diag}(\hat{\mathbf{f}}_h)/c\\
    \boldsymbol{0}&\boldsymbol{0}
\end{pmatrix}\begin{pmatrix}
    \hat{\boldsymbol{\mathbf{U}}}_h\\
    \boldsymbol{r}
\end{pmatrix},\quad 
\begin{pmatrix}
    \hat{\boldsymbol{\mathbf{U}}}_h(0)\\
    \boldsymbol{r}(0)
\end{pmatrix} = \begin{pmatrix}
     (I\otimes \Phi)\mathbf{U}_h(0)\\
     \boldsymbol{r}_0
\end{pmatrix},
\end{equation}
where $\boldsymbol{r}_0=c\sum_{i \in[9M^3]}\ket{j}, c \sim \|\hat{\mathbf{f}}_h\|_{\infty}$. Simplify \eqref{SMF_spectral_ode} as:
\begin{equation}\label{SMF_spectral_ode_finial}
\frac{d}{dt} \boldsymbol{u} =  H_1\boldsymbol{u}+iH_2\boldsymbol{u}.\quad 
\boldsymbol{u}(0) = \begin{pmatrix}
     \hat{\boldsymbol{\mathbf{U}}}_h(0)\\
     \boldsymbol{r}(0)\\
     \boldsymbol{1}_{12M^3\times 1}
\end{pmatrix}.
\end{equation}
The addition of $\boldsymbol{1}$ to the vector ensures that the dimension aligns with a power of $2$, enabling efficient encoding into quantum qubits. The matrices $H_1$ and  $H_2$ can be written as follows:
\begin{equation}\label{SMF_H1H2}
H_1 = \frac{1}{2}\begin{pmatrix}
    \boldsymbol{0}&\operatorname{diag}(\hat{\mathbf{f}}_h)/c&\boldsymbol{0}\\
    \operatorname{diag}(\hat{\mathbf{f}}_h^{\dagger})/c&\boldsymbol{0}&\boldsymbol{0}\\
    \boldsymbol{0}&\boldsymbol{0}&\boldsymbol{0}
\end{pmatrix},
H_2 = \frac{1}{2i}\begin{pmatrix}
    2A&\operatorname{diag}(\hat{\mathbf{f}}_h)/c&\boldsymbol{0}\\
    -\operatorname{diag}(\hat{\mathbf{f}}_h^{\dagger})/c&\boldsymbol{0}&\boldsymbol{0}\\
    \boldsymbol{0}&\boldsymbol{0}&\boldsymbol{0}
\end{pmatrix}.
\end{equation}
The dimension of zero matrices can be inferred from the context, so it is omitted.

At this point, we can apply the Schr$\ddot{\text{o}}$dingerisation method on Eq. \eqref{SMF_spectral_ode_finial}. Taking the warped transformation and then taking the discrete Fourier transform of the new variable $p$, one can get the following Hamiltonian systems:
\begin{subequations}\label{SMF_spectral_Schr}
    \begin{align}
        \frac{d}{dt} \boldsymbol{c} &= -i[(H_1\otimes D_p)-(H_2\otimes I)] \boldsymbol{c}=-iH_s^d \boldsymbol{c},\label{vs_spectral_Schr_d}
    \end{align}
\end{subequations}
where $H_1$ and $H_2$ are given by \eqref{SMF_H1H2}. $D_p$ is given by \eqref{fourier_spectral_D}. The above process can be summarized as follows:
    \begin{equation*}
     \boldsymbol{\mathbf{U}}_h \xrightarrow{\text { DFT on } x} \hat{\boldsymbol{\mathbf{U}}}_h \xrightarrow{\text{ homogeneous}} \boldsymbol{u}\xrightarrow[{v = g(p)\boldsymbol{u}}]{\text { warped trasformation }} v \xrightarrow{\text { discrete on } p} \boldsymbol{v}_h \xrightarrow{\text { DFT on } p} \boldsymbol{c},
    \end{equation*}
One can obtain the relations:
\begin{subequations}
    \begin{align}
        \boldsymbol{c}&=(\boldsymbol{\Phi}_{x}^{-1}\boldsymbol{\mathbf{U}}_h)\otimes(\Phi_p^{-1}\sum_{j=0}^{N-1} e^{-|p_j|}\ket{j}).
    \end{align}
\end{subequations}

\begin{theorem}\label{thm: SMFeqcomplexity}
In the $d$-dimensional case, given the initial quantum state $\ket{\boldsymbol{u}(0)}$, assume that the spatial grid number and extended $p$ grid number are $M = 2^m$ and $N=2^n$ respectively, the solution is smooth
enough such that $\boldsymbol{\mathbf{U}} \in H^r(\Omega)$, $ r \geq 2$. Over the evolution time $T$, the gate complexity of quantum simulation of \eqref{SMF_spectral_Schr} with success probability at least $1-2\epsilon$ is 
    \begin{equation}\label{w_spectral_DFT_complexity}
        N_{Gate}=  \mathscr{O}((\lceil\log{(d^2+3d)}\rceil+\frac{d}{r}\log\frac{1}{\epsilon})\epsilon^{-1}T).
    \end{equation}
\end{theorem}

\begin{proof}
In spectral method, the error depends on the smoothness of the solution. Therefore, there is $r$-order accuracy in the $x$-direction which implies
\begin{equation}
    M=\mathscr{O}(\epsilon^{-\frac{1}{r}}).
\end{equation}
In the $p$-direction, since the function is only continuous but not continuously differentiable, the total number of discrete lattice points is 
\begin{equation}
N=\mathscr{O}({\epsilon}^{-1}).
\end{equation}
The Hamiltonians $H_s^d$ in Eqs. \eqref{SMF_spectral_Schr} inherits the sparsity of $A$, so $s =\sum_{i=1,2,3}s(L_i) = 3$. Hence 
\begin{gather*}
    \|H_s^d\|_{\max}=\max\{\|H_1\otimes D_p\|_{\max},\|H_2\|_{\max}\}=\max\{\mathscr{O}(N),\mathscr{O}(M)\} = \mathscr{O}(\epsilon^{-1}).
\end{gather*}

So one have $N_{query} = \mathscr{O}(s\epsilon^{-1}T+\frac{\log{\epsilon^{-1}}}{\log{\log{\epsilon^{-1}}}})$ by Lemma~\ref{lem:ham_simu complexity}. The number of qubits is 
\begin{gather*}
    m_{H_s^d} = \lceil \log{[(d^2+3d)M^dN]}\rceil = \lceil \log{(d^2+3d)}\rceil+(\frac{d}{r}+1) \mathscr{O}(\log \frac{1}{\epsilon}),
\end{gather*}
where $\frac{d^2+3d}{2}$ is the dimension of the SMF and the precision of matrix elements satisfies $\epsilon = 2^{-m_b}$. The gate complexity of the Hamiltonian simulation is
\begin{gather*}
    N_{Gate}=  \mathscr{O}[s\epsilon^{-1}T+\frac{\log\frac{1}{\epsilon}}{\log{\log\frac{1}{\epsilon}}}][\log{(d^2+3d)}+(\frac{d}{r}+1)\log\frac{1}{\epsilon}+\log\frac{1}{\epsilon}\operatorname{polylog}\log\frac{1}{\epsilon}]
\end{gather*}
In addition, the complexity of the Quantum Fourier Transform is $\mathscr{O}(n\log n)$\cite{QFT}. This concludes the proof.
\end{proof}
\begin{remark}
When using a smooth initialization function $g(p)\in C^k(\mathbb{R})$ as \eqref{g(p)}, the complexity will reduce to $d\log{\frac{1}{\epsilon}}\cdot\max\{\mathscr{O}(\epsilon^{-\frac{1}{k}}T),\mathscr{O}(\epsilon^{-\frac{1}{r}}T)\}$.
\end{remark}
\begin{remark}\label{rem:SMF_classical_complexity}
     For Eqs. \eqref{v-s_SMF}, classical implementation requires $\mathscr{O}((\frac{d^2+3d}{2})M^d)$ memories to store the discretized field $\mathbf{U}_h$. Considering the same spectral method,  the complexity of the Fast Fourier Transform is $\mathscr{O}(M^d\log{M^d})$ referring to Eq. \eqref{SMF_spectral_D}, so the application of the operator $A$ to the vector $\hat{\mathbf{U}}_h$ requires $\mathscr{O}((\frac{d^2+3d}{2})sM^d\log{M^d})$ arithmetic operations, where $s = 3$. If one uses the Crank-Nicolson scheme in time, classical simulation up to time $T$ within the additive error $\epsilon$ requires $\mathscr{O}(T^2/\epsilon)$ steps, which results in
     \begin{align*}
          \mathscr{O}((\frac{d^2+3d}{2})sM^d\log{M^d}T^2/\epsilon)
          = \mathscr{O}(\frac{(d^2+3d)d}{2r}\epsilon^{-(\frac{d}{r}+1)}\log{\frac{1}{\epsilon}}T^2)
     \end{align*}arithmetic operations. As shown in Theorem \ref{thm: SMFeqcomplexity}, when $T\sim\mathscr{O}(1)$, Schr$\ddot{\text{o}}$dingerisation method has exponential advantage in $d$.
\end{remark}

\subsection{Quantum simulation of velocity-stress equation with variable coefficients by staggered grid method}\label{sec_vaco_staggered}

In seismology and materials science, the properties of the medium through which elastic waves propagate typically vary with position. In this section, we will focus on the propagation of the first-order elastic wave equations \eqref{velocity-stress} in inhomogeneous media, where the elastic parameters depend on spatial variables.

Among the classical algorithms, the staggered grid method is widely used to solve the velocity-stress equation\cite{1986P,Randall}. Staggered grid refers to a grid system in which velocity and stress are defined on two different grids. This method can well deal with the coupling relationship between velocity and stress in the first-order wave equations. We use a uniform  spatial grid number $M_x = M_y = M_z = M$. Different variables are approximated at different points as:
\begin{align*}
&\boldsymbol{v_1} = \sum_{i,j,k} v_1(x_{i+\frac{1}{2}},y_j,z_k) \ket{i}\ket{j}\ket{k}, \quad
\boldsymbol{v_2} = \sum_{i,j,k} v_2(x_i,y_{j+\frac{1}{2}},z_k) \ket{i}\ket{j}\ket{k}, \\
&\boldsymbol{v_3} = \sum_{i,j,k} v_3(x_i,y_j,z_{k+\frac{1}{2}}) \ket{i}\ket{j}\ket{k}; \\
&\boldsymbol{\sigma_{11}} = \sum_{i,j,k} \sigma_{11}(x_i,y_j,z_k)\ket{i}\ket{j}\ket{k}, \quad
\boldsymbol{\sigma_{22}} = \sum_{i,j,k} \sigma_{22}(x_i,y_j,z_k)\ket{i}\ket{j}\ket{k}, \\
&\boldsymbol{\sigma_{33}} = \sum_{i,j,k} \sigma_{33}(x_i,y_j,z_k)\ket{i}\ket{j}\ket{k}; \\
&\boldsymbol{\sigma_{12}} = \sum_{i,j,k} \sigma_{12}(x_{i+\frac{1}{2}},y_{j+\frac{1}{2}},z_k)\ket{i}\ket{j}\ket{k}, \quad
\boldsymbol{\sigma_{13}} = \sum_{i,j,k} \sigma_{13}(x_{i+\frac{1}{2}},y_{j},z_{k+\frac{1}{2}})\ket{i}\ket{j}\ket{k}, \\
&\boldsymbol{\sigma_{23}} = \sum_{i,j,k} \sigma_{23}(x_{i},y_{j+\frac{1}{2}},z_{k+\frac{1}{2}})\ket{i}\ket{j}\ket{k}, \quad j_1,j_2,j_3 \in [ M ].
\end{align*}
Denote the collections of velocity and stress as:
$$
\boldsymbol{v} = (\boldsymbol{v_1},\boldsymbol{v_2},\boldsymbol{v_2}),\quad \boldsymbol{\sigma} = (\boldsymbol{\sigma_{11}},\boldsymbol{\sigma_{22}},\boldsymbol{\sigma_{33}},\boldsymbol{\sigma_{12}},\boldsymbol{\sigma_{13}},\boldsymbol{\sigma_{23}} ),\quad \boldsymbol{u} = (\boldsymbol{v},\boldsymbol{\sigma}).
$$
Define the matrices in the following way:
\begin{equation}\label{Smatrix}
S = \frac{1}{h}\sum_{i=0}^{M-2} \ket{i}\bra{i+1}+\ket{M-1}\bra{0}-\boldsymbol{1}_M,\quad
S_i = I^{\otimes^{i-1}} \otimes S \otimes I^{\otimes^{3-i}}.
\end{equation}
Using the staggered grid method to discretize Eq. \eqref{velocity-stress}, we can obtain:
\begin{equation}\label{v_s_stagger_discrete}
\frac{d}{dt} \boldsymbol{u} = A_H\boldsymbol{u}, \quad A_H = \begin{pmatrix}
\boldsymbol{0} & R^{-1}\boldsymbol{L}_v\\
-C\boldsymbol{L}_v^T& \boldsymbol{0}
\end{pmatrix},\quad
\boldsymbol{u}(0) = \begin{pmatrix}
     \boldsymbol{v}(0)\\
     \boldsymbol{\sigma}(0)\
     \end{pmatrix},
\end{equation}
where $R \in \mathbb{R}^{3M^3\times 3M^3}$ and $C\in \mathbb{R}^{6M^3\times 6M^3}$ collect  the elasticity parameters at the corresponding cell boundary or cell center.
$\boldsymbol{L}_{v}\in \mathbb{R}^{3M^3\times 6M^3}$ are the staggered grid difference matrices:
\begin{equation}
\boldsymbol{L}_v  = \begin{bmatrix}
S_1  &  &  &  -S^T_2&  -S^T_3& \\
&S_2  &  &-S^T_1  &  &-S^T_3 \\
&  & S_3 &  & -S^T_1&-S^T_2
\end{bmatrix}.
\end{equation}
\subsubsection{The Schr$\ddot{\text{o}}$dingerisation Simulation}\label{sec:vsSchrSim}

One can decompose the evolution operator $A_H$ of Eq. \eqref{v_s_stagger_discrete} into the Hermian and anti-Hermitian parts:
\begin{equation}\label{vseq_H1H2}
\frac{d}{dt} \boldsymbol{u} =  H_1\boldsymbol{u}+iH_2\boldsymbol{u},\quad 
\boldsymbol{u}(0) = \begin{pmatrix}
     \boldsymbol{v}(0)\\
     \boldsymbol{\sigma}(0)\\
     \boldsymbol{1}
\end{pmatrix}.
\end{equation}
The matrices $H_1$ and  $H_2$ can be written as follows:
\begin{equation}\label{vs_H1H2}
H_1 = \frac{1}{2}\begin{pmatrix}
    \boldsymbol{0}&R^{-1}\boldsymbol{L}_v-\boldsymbol{L}_vC&\boldsymbol{0}\\
    -C\boldsymbol{L}_v^T+\boldsymbol{L}_v^TR^{-1}&\boldsymbol{0}&\boldsymbol{0}\\
    \boldsymbol{0}&\boldsymbol{0}&\boldsymbol{0}
\end{pmatrix},
H_2 = \frac{1}{2i}\begin{pmatrix}
    \boldsymbol{0}&R^{-1}\boldsymbol{L}_v+\boldsymbol{L}_vC&\boldsymbol{0}\\
    -C\boldsymbol{L}_v^T-\boldsymbol{L}_v^TR^{-1}&\boldsymbol{0}&\boldsymbol{0}\\
    \boldsymbol{0}&\boldsymbol{0}&\boldsymbol{0}
\end{pmatrix}.
\end{equation}

One can obtain the following Hamiltonian systems by Schr$\ddot{\text{o}}$dingerisation method:
\begin{subequations}\label{vs_staggered_Schr}
    \begin{align}
        \frac{d}{dt} \boldsymbol{c} &= -i[(H_1\otimes D_p)-(H_2\otimes I)] \boldsymbol{c}=-iH_s^d \boldsymbol{c},\label{vs_staggered_Schr_d}
    \end{align}
\end{subequations}
where $H_1$ and $H_2$ are given by \eqref{vs_H1H2}. The matrices $D_p$ is the same as Sec. \ref{sec:vs1SchrSim}. Different variables have the relationship:
\begin{subequations}
    \begin{align}
        \boldsymbol{c}&=(I\otimes\Phi_p^{-1})(\boldsymbol{u}\otimes\sum_{j=0}^{N-1} {e^{-|p_j|}}\ket{j})),
    \end{align}
\end{subequations}
\begin{theorem}\label{thm: vseqcomplexity}
Under the same conditions as in Theorem \ref{thm: SMFeqcomplexity}, over the evolution time $T$, the gate complexity of quantum simulation of \eqref{vs_staggered_Schr} with success probability at least $1-2\epsilon$ is 
    \begin{equation}\label{w_spectral_DFT_complexity}
        N_{Gate}=  \mathscr{O}((\lceil\log(d^2+3d)\rceil+\frac{d}{2}\log\frac{1}{\epsilon})\epsilon^{-\frac{3}{2}}T).
    \end{equation}
\end{theorem}
\begin{proof}
Because of the second-order accuracy in the $x$-direction and the lack of regularity in the $p$-direction, the total number of discrete lattice points is 
\begin{equation}
    M=\mathscr{O}(\epsilon^{-\frac{1}{2}}), \quad N=\mathscr{O}({\epsilon}^{-1}).
\end{equation}
 The sparsity of Hamiltonians $H_s^d$ in \eqref{vs_staggered_Schr} is $s =3s(S)=6$. And
\begin{gather*}
    \|H_s^d\|_{\max}=\max\{\mathscr{O}(MN),\mathscr{O}(M)\} = \mathscr{O}(\epsilon^{-\frac{3}{2}}),
\end{gather*}
The remaining proof is similar to Theorem \ref{thm: SMFeqcomplexity}, so we omit it here.
\end{proof}
\begin{remark}
When using a smooth initialization function $g(p)\in C^k(\mathbb{R})$ as \eqref{g(p)}, the complexity will reduce to $\mathscr{O}(d\log{\frac{1}{\epsilon}}\epsilon^{-(\frac{1}{2}+\frac{1}{k})}T)$.
\end{remark}
\begin{remark}
     For Eqs. \eqref{eq5}, classical computation requires $\mathscr{O}((\frac{d^2+3d}{2})M^d)$ memories to store the discretized field $\boldsymbol{u}$. Using the same staggered grid spatial discretization scheme, the difference operator needs $\mathscr{O}(s(\frac{d^2+3d}{2})M^d)$ arithmetic operations, where $s = 3s(S)=6$. If one uses the Crank-Nicolson scheme in time, classical simulation up to time $T$ within the additive error $\epsilon$ requires $\mathscr{O}(T^2/\epsilon)$ steps, which results in
     \begin{align*}
          \mathscr{O}((\frac{d^2+3d}{2})sM^dT^2/\epsilon)
          = \mathscr{O}((\frac{d^2+3d}{2})\epsilon^{-(\frac{d}{2}+1)}T^2)
     \end{align*}arithmetic operations. As shown in  Theorem \ref{thm: vseqcomplexity}, the Schr$\ddot{\text{o}}$dingerisation method has an exponential advantage in $d$. 
\end{remark}

\section{Quantum simulation of elastic wave displacement equations in isotropic media}\label{sec:QS1}

In this section, we consider the quantum simulation of elastic wave equations \eqref{Tran1} via the Schr$\ddot{\text{o}}$dingerisation method. Using the spectral method and the central difference scheme for spatial discretizations, we adopt a uniform spatial grid size $\Delta x = \Delta y = \Delta z = \frac{b-a}{M}$, $M = 2^m$.

\subsection{Quantum simulation of elastic wave equations \eqref{Tran1} by spectral method}\label{sec_2_spectral}

 The spatial domain and grid points will be the same as in Sec \ref{sec:SMFofvseq}, 
\begin{equation}
    \boldsymbol{\mathbf{w}}_{i,\boldsymbol{j}} = \mathbf{w}_i(t,\boldsymbol{x_j}), \quad \mathbf{w}_h = \sum_{i=0}^{9}\sum_{j_1,j_2,j_3 = 0}^m \mathbf{w}_i(t,\boldsymbol{x_j})\ket{i}.
\end{equation}
The expressions of the external force  function at the lattice points are
\begin{equation}
    \boldsymbol{f}_{i,\boldsymbol{j}} = f_i(t,\boldsymbol{x_j}), \quad
\boldsymbol{f}_h = \sum_{i=0}^{2}\sum_{j_1,j_2,j_3 = 0}^mf_i(t,\boldsymbol{x_j})\ket{i}.
\end{equation}
Take the discrete Fourier transform of the spatial variables of \eqref{eq5}, one gets
\begin{equation}\label{secondoerder_EWE_spectral_P}
    \frac{d}{dt} \mathbf{w}_h +L_{\mu} \mathbf{w}_h+\boldsymbol{f}_h=0,
\end{equation}
where 
\begin{equation}
    L_{\mu} = -i\left(\begin{array}{cccc}
0 & K_{\mu} / \rho & M_{\mu} / \rho & \frac{1}{\rho} \boldsymbol{k}_{\mu} \\
2 \mu K_{\mu} & 0 & 0 & 0 \\
\mu M_{\mu} & 0 & 0 & 0 \\
\lambda \boldsymbol{k}_{\mu}^t & 0 & 0 & 0
\end{array}\right),
\end{equation}
with  corresponding momentum operators 
\begin{equation}
    \boldsymbol{k}_{\mu} = 
\left(\begin{array}{c}
\boldsymbol{P}_1\\
\boldsymbol{P}_2\\
\boldsymbol{P}_3
\end{array}\right),\quad
M_{\mu}=\left(\begin{array}{ccc}
0 & \boldsymbol{P}_3 & \boldsymbol{P}_2 \\
\boldsymbol{P}_3 & 0 & \boldsymbol{P}_1 \\
\boldsymbol{P}_2 & \boldsymbol{P}_1 & 0
\end{array}\right),\quad
K_{\mu} = \left(\begin{array}{ccc}
\boldsymbol{P}_1& 0 & 0 \\
0 & \boldsymbol{P}_2 & 0 \\
0&0 & \boldsymbol{P}_3
\end{array}\right),
\end{equation}
where $\boldsymbol{P}_i = I^{\otimes^{i-1}}\otimes P_{\mu} \otimes I^{\otimes^{3-i}}$. Let $
\hat{\boldsymbol{w}}_h = (I\otimes \boldsymbol{\Phi}^{-1}) \mathbf{w}_h,  
\hat{\boldsymbol{f}} = (I\otimes \boldsymbol{\Phi}^{-1}) \boldsymbol{f}_h, 
\boldsymbol{\Phi}=\Phi^{\otimes 3}$ , we can also reduce the sparsity and write the equaion \eqref{secondoerder_EWE_spectral_P} as: 
\begin{equation}\label{secondoerder_EWE_spectral_D}
    \frac{d}{dt}\hat{\boldsymbol{w}}_h+L_s \hat{\boldsymbol{w}}_h+\hat{\boldsymbol{f}}_h=0,
\end{equation}
where 
\begin{equation}
    L_s = -i\left(\begin{array}{cccc}
0 & K_s / \rho & M_s / \rho & \frac{1}{\rho} \boldsymbol{k}_s \\
2 \mu K_s & 0 & 0 & 0 \\
\mu M_s & 0 & 0 & 0 \\
\lambda (\boldsymbol{k}_s)^t & 0 & 0 & 0
\end{array}\right),
\end{equation}
while the momentum matrices become
\begin{equation}
    \boldsymbol{k}_s = 
\left(\begin{array}{c}
\boldsymbol{D}_1\\
\boldsymbol{D}_2\\
\boldsymbol{D}_3
\end{array}\right),\quad
M_s=\left(\begin{array}{ccc}
0 & \boldsymbol{D}_3 & \boldsymbol{D}_2 \\
\boldsymbol{D}_3 & 0 & \boldsymbol{D}_1 \\
\boldsymbol{D}_2 & \boldsymbol{D}_1 & 0
\end{array}\right),\quad
K_s = \left(\begin{array}{ccc}
\boldsymbol{D}_1& 0 & 0 \\
0 & \boldsymbol{D}_2 & 0 \\
0&0 & \boldsymbol{D}_3
\end{array}\right).
\end{equation}
Correspondingly, $\boldsymbol{D}_i$ satisfies 
$\boldsymbol{D}_i = I^{\otimes^{i-1}}\otimes D_{\mu} \otimes I^{\otimes^{3-i}}$.
\subsubsection{The Schr$\ddot{\text{o}}$dingerisation Simulation}

Applying the Schr\"odingerisation, the corresponding homogeneous system is : 
\begin{equation}\label{spectral_ode_finial}
\frac{d}{dt} \boldsymbol{u} =  H_1\boldsymbol{u}+iH_2\boldsymbol{u}.\quad 
\boldsymbol{u}(0) = \begin{pmatrix}
     \hat{\boldsymbol{w}}_h(0)\\
     \boldsymbol{r}(0)\\
     \boldsymbol{1}_{12M^3\times 1}
\end{pmatrix}
\end{equation}
The matrices $H_1$ and $H_2$ are defined by:
\begin{equation}\label{w_specrtal_H1H2}
H_1 = \begin{pmatrix}
    H_{s1}&\frac{1}{2}D(\tilde{\boldsymbol{f}})&\boldsymbol{0}\\
    \frac{1}{2}D(\tilde{\boldsymbol{f}})&\boldsymbol{0}&\boldsymbol{0}\\
    \boldsymbol{0}&\boldsymbol{0}&\boldsymbol{0}
\end{pmatrix},\quad 
H_2 = \begin{pmatrix}
    H_{s2}&\frac{1}{2i}D(\tilde{\boldsymbol{f}})&\boldsymbol{0}\\
    -\frac{1}{2i}D(\tilde{\boldsymbol{f}})&\boldsymbol{0}&\boldsymbol{0}\\
    \boldsymbol{0}&\boldsymbol{0}&\boldsymbol{0}
\end{pmatrix},
\end{equation}
where $H_{s1} = \frac{1}{2}[(L_s+L_s^\dagger] = \frac{i}{2}H_{-}$ and $H_{s2}=\frac{1}{2i}[(L_s-L_s^\dagger] = \frac{1}{2}H_{+}$, with
\begin{equation}\label{H_pm}
   H_{\pm} = \left(\begin{array}{cccc}
0 &(\frac{1}{ \rho} \pm2\mu)K_{s} & (\frac{1}{ \rho} \pm \mu)M_{s} & (\frac{1}{ \rho} \pm\lambda)\boldsymbol{k}_{s} \\
(2\mu \pm\frac{1}{ \rho}) K_{s} & 0 & 0 & 0 \\
(\mu \pm\frac{1}{ \rho}) M_{s} & 0 & 0 & 0 \\
(\lambda\pm\frac{1}{ \rho})\boldsymbol{k}_{s}^t & 0 & 0 & 0
\end{array}\right).
\end{equation}

The Hamiltonian system obtained through  Schr$\ddot{\text{o}}$dingerisation method is 
\begin{subequations}\label{w_spectral_Schr}
    \begin{align}
        \frac{d}{dt} \boldsymbol{c} &= -i[(H_1\otimes D_p)-(H_2\otimes I)] \boldsymbol{c}=-iH_s^d \boldsymbol{c},\label{w_spectral_Schr_d}
    \end{align}
\end{subequations}
where $H_1$ and $H_2$ are given in \eqref{w_specrtal_H1H2}.
  

\begin{theorem}\label{thm: spectralcomplexity}
Under the same conditions as in Theorem \ref{thm: SMFeqcomplexity}, in $d$-dimensional case, the solution is smooth enough such that $\boldsymbol{w} \in H^r(\Omega), r \geq 2$. The gate complexity of quantum simulation of \eqref{w_spectral_Schr} with success probability at least $1-2\epsilon$ is 
\begin{equation}\label{w_spectral_complexity}
        N^d_{Gate}=  \mathscr{O}[(\lceil\log{(\frac{d^2+3d}{2}+1)}\rceil+\frac{d}{r}\log \frac{1}{\epsilon})\epsilon^{-(\frac{1}{r}+1)}T].
    \end{equation}
\end{theorem}

\begin{proof}
Because of $\boldsymbol{w}\in H^r(\Omega)$ and $e^{-|p|}\in C^0(\mathbb{R}) \setminus C^1(\mathbb{R})$, as shown in Theorem. \ref{thm: SMFeqcomplexity}, one have
\begin{equation}
    M=\mathscr{O}(\epsilon^{-\frac{1}{r}}),\quad N=\mathscr{O}(\epsilon^{-1})
\end{equation}
And the dimension of the varivales in Eqs. \eqref{eq5} is $\frac{d^2+3d}{2}+1$. The remaining proof is similar to \ref{thm: SMFeqcomplexity}, so we omit it here.
\end{proof}

\begin{remark}
When using a smooth initialization function $g(p)\in C^k(\mathbb{R})$ as \eqref{g(p)}, the complexity will reduce to $\mathscr{O}(d\log{\frac{1}{\epsilon}}\epsilon^{-(\frac{1}{r}+\frac{1}{k})}T)$.
\end{remark}

\begin{remark}
     For Eqs. \eqref{eq5}, the arithmetic operations required for spatially spectral method and temporally Crank-Nicolson scheme is 
     \begin{align*}
          \mathscr{O}((\frac{d^2+3d}{2}+1)\epsilon^{-(\frac{d}{r}+1)}T^2 ).
     \end{align*} It can still be observed that Schr$\ddot{\text{o}}$dingerisation method has exponential advantage when $T\sim\mathscr{O}(1)$.
\end{remark}

\subsection{Schr$\ddot{\text{o}}$dingerisation with central difference scheme}

The same mesh discretisation and function symbols are defined  for the $3D$ case as in Sec.  \ref{sec_2_spectral}. One can apply the central difference scheme in $\boldsymbol{x}$-direction. Due to the periodic boundary conditions, denote $D$ the central difference matrix in 1-D:
\begin{equation}\label{central_D}
    D = \sum_{i=0}^{M-1} (\ket{i}\bra{i+1}+\ket{i+1}\bra{i})-\ket{0}\bra{M-1}+\ket{M-1}\bra{0}.
\end{equation}
Approximating Eq. \eqref{eq5} by the central difference scheme, one can obtain
\begin{equation}\label{w_central_ode}
    \frac{d}{dt} \mathbf{w}_h +L_c \mathbf{w}_h + \boldsymbol{f}_h = 0,
\end{equation}
where $L_c$ has the same form as $L_s$ in Eq. \eqref{secondoerder_EWE_spectral_D}, but the element $\boldsymbol{D}_i$ in $L_c$ satisfy:
\begin{equation}
    \boldsymbol{D}_i = I^{\otimes^{i-1}}\otimes D \otimes I^{\otimes^{3-i}}.
\end{equation}
\subsubsection{Schr$\ddot{\text{o}}$dingerisation Simulation}

Following  the same steps as Sec. \ref{sec_2_spectral}, we  obtain the homogeneous ODEs:
\begin{equation}\label{central_ode_H}
    \frac{d}{dt} \boldsymbol{u} =  H_1\boldsymbol{u}+iH_2\boldsymbol{u},\quad 
    \boldsymbol{u}(0) = \begin{pmatrix}
     \mathbf{w}_h(0)\\
     \boldsymbol{r}(0)\\
     \boldsymbol{1}_{12M^3\times1}
\end{pmatrix}.
\end{equation}
The matrices $H_1$ and $H_2$ have the following forms:
\begin{equation}\label{w_cenrtal_H1H2}
H_1 = \begin{pmatrix}
    H_{c1}&\frac{1}{2}D(\tilde{\boldsymbol{f}})&\boldsymbol{0}\\
    \frac{1}{2}D(\tilde{\boldsymbol{f}})&\boldsymbol{0}&\boldsymbol{0}\\
    \boldsymbol{0}&\boldsymbol{0}&\boldsymbol{0}
\end{pmatrix},\quad 
H_2 = \begin{pmatrix}
    H_{c2}&\frac{1}{2i}D(\tilde{\boldsymbol{f}})&\boldsymbol{0}\\
    -\frac{1}{2i}D(\tilde{\boldsymbol{f}})&\boldsymbol{0}&\boldsymbol{0}\\
    \boldsymbol{0}&\boldsymbol{0}&\boldsymbol{0}
\end{pmatrix},
\end{equation}
where $H_{c1} = \frac{1}{2}H_{-}$ and $H_{c2}=\frac{1}{2i}H_{+}$, requiring replacing $M_s$, $K_s$ and $\boldsymbol{k}_s$ with $M_c$, $K_c$ and $\boldsymbol{k}_c$. One can get the corresponding Hamiltonian system: 
\begin{subequations}\label{w_central_Schr}
    \begin{align}
        \frac{d}{dt} \boldsymbol{c} &= -i[(H_1\otimes D_p)-(H_2\otimes I)] \boldsymbol{c}=-iH^d_c \boldsymbol{c},\label{w_central_Schr_d}
    \end{align}
\end{subequations}
where $H_1$ and $H_2$ are given by \eqref{w_cenrtal_H1H2}. $D_p$ and $D_{\xi}$ are same as \eqref{w_spectral_Schr}.

\begin{theorem}\label{thm: centralcomplexity}
Under the same conditions as in Theorem \ref{thm: spectralcomplexity}, in $d$-dimensional case, the gate complexity of \eqref{w_central_Schr} is 
\begin{equation}\label{w_central_DFT_complexity}
    \begin{split}
        N^d_{Gate}=  \mathscr{O}((\lceil\log{(\frac{d^2+3d}{2}+1)}\rceil+\frac{d}{2}\log \frac{1}{\epsilon})\epsilon^{-\frac{3}{2}}T).
    \end{split}
    \end{equation}
\end{theorem}
\begin{proof}
Because of the second-order accuracy in the $x$-direction, the total number of discrete  points is $M=\mathscr{O}((\frac{1}{\epsilon})^{\frac{1}{2}})$. The sparsity of $L_c$ also becomes $s(L_c)=4\,s(D)+1=9$.
The proof is similar to that of Theorem \ref{thm: vseqcomplexity}.
\end{proof}
\begin{remark}
Using the Crank-Nicolson scheme, the total number of arithmetic operations required for classical simulation up to time $T$ is $\mathscr{O}((\frac{d^2+3d}{2}+1)\epsilon^{-(\frac{d}{2}+1)}T^2)$. Clearly for $d \geq 2$, the Schr$\ddot{\text{o}}$dingerisation method has exponential advantage in $d$.
\end{remark}

\begin{theorem}
Assuming that the medium parameters are $\mathscr{O}(1)$, the numbers of space points for Eqs. \eqref{w_central_Schr} and \eqref{w_spectral_Schr} are  $M^s$ and $M^c$, respectively. Then the corresponding $p^*$ in the Schr$\ddot{\text{o}}$dingerisation framework required to recover the solution are respectively:
\begin{equation}
    p_s^* \sim \mathscr{O}\left(\left(\frac{1}{\epsilon}\right)^{\frac{1}{r}}\right ),\quad\text{with spectral method},
\end{equation}
\begin{equation}
    p_c^* \sim \mathscr{O}\left(\left(\frac{1}{\epsilon}\right)^{\frac{1}{2}}\right),\quad\text{with central difference method}.
\end{equation}
Therefore, the former selects a larger domain of the auxiliary  variable $p$ than the latter.
\end{theorem}
\begin{proof}
Define the matrix
\begin{equation}\label{parameter_matrix}
P = \left(\begin{array}{cccc}
0 &(\frac{1}{ \rho} -2\mu) & (\frac{1}{ \rho} - \mu)  & (\frac{1}{ \rho} -\lambda)  \\
(2\mu -\frac{1}{ \rho})   & 0 & 0 & 0 \\
(\mu -\frac{1}{ \rho})   & 0 & 0 & 0 \\
(\lambda-\frac{1}{ \rho}) ^t & 0 & 0 & 0
\end{array}\right).
\end{equation}
It is easy to calculate the eigenvalues of the Hermite matrix $H_1$ in Schr$\ddot{\text{o}}$dingerisation,
\begin{align}
     \max\{|\lambda_n( H^s_1)|\} &= \max\{|\lambda_n( P)|\}\pi M^s/2,\\
     \max\{|\lambda_n( H^c_1)|\} &= \max\{|\lambda_n( P)|\}M^c/2.
\end{align}
\end{proof}

\begin{remark}
Eq. \eqref{eq5} can be decomposed into the form of the following equations depending on the spatial derivatives in different directions:
    \begin{equation}
        \frac{\partial \mathbf{w}}{\partial t}+A\mathbf{w}_x + B\mathbf{w}_y+C\mathbf{w}_z+\boldsymbol{F} =0.
    \end{equation}
    The non-commutativity of $A, B$, and $C$ poses challenges for traditional upwind schemes. However, the operator splitting method \cite{StrangSplitting} can be effectively applied to address this issue. While combining operator splitting with Schr$\ddot{\text{o}}$dingerisation increases the complexity of the implementation, it is important to note that the resulting ODEs with asymmetric coefficient matrices are well-suited for Schr$\ddot{\text{o}}$dingerisation, which also offers significant advantages in this context. Thus, we elaborate on the implementation of combining the central difference scheme with Schr$\ddot{\text{o}}$dingerisation. 
\end{remark}

\section{Numerical experiments}\label{sec:Numricalsimulation}

In this section, we perform the numerical experimentations to verify the correctness of the Schr\"odingerzation method.   All numerical tests are carried out on a classical computer. Since the quantum computers are not available, one cannot verify the complexity at the moment.
\subsection{One-dimensional velocity-stress equations with source terms}
In this numerical experiment, we consider the velocity and stress fields for waves propagating through a homogeneous medium under constant external forcing $f = 0.1$. The initial velocity is zero and the initial stress is given by the Gaussian function:
\begin{equation}
    \sigma(x,0) = \exp^{-(x-5)^2}.
\end{equation}
We use the spectral method for quantum simulations in Sec \ref{sec:SMFofvseq} and the temporal scheme is the Crank-Nicolson scheme with time step $\Delta t=\frac{1}{100}$. When the spatial grid number is set to $N_x=2^6$, the maximum eigenvalue of $H_1$ in Eq. \eqref{SMF_spectral_ode_finial} is $3.200$. In the discrete framework, the $p$ domain is chosen to be $[-4.2,5]$ with $N_p = 2^{10}$, so we pick $p_1 = 3.203$ to compute the one point recovery. Figure \ref{fig:vaco} shows the comparison between the solutions obtained by the Schr$\ddot{\text{o}}$dingerisation method and the classical spectral method. They are in good agreement.

\begin{figure}[htbp]
    \centering
    \subfigure[The velocity field]{
        \includegraphics[width=0.4\linewidth]{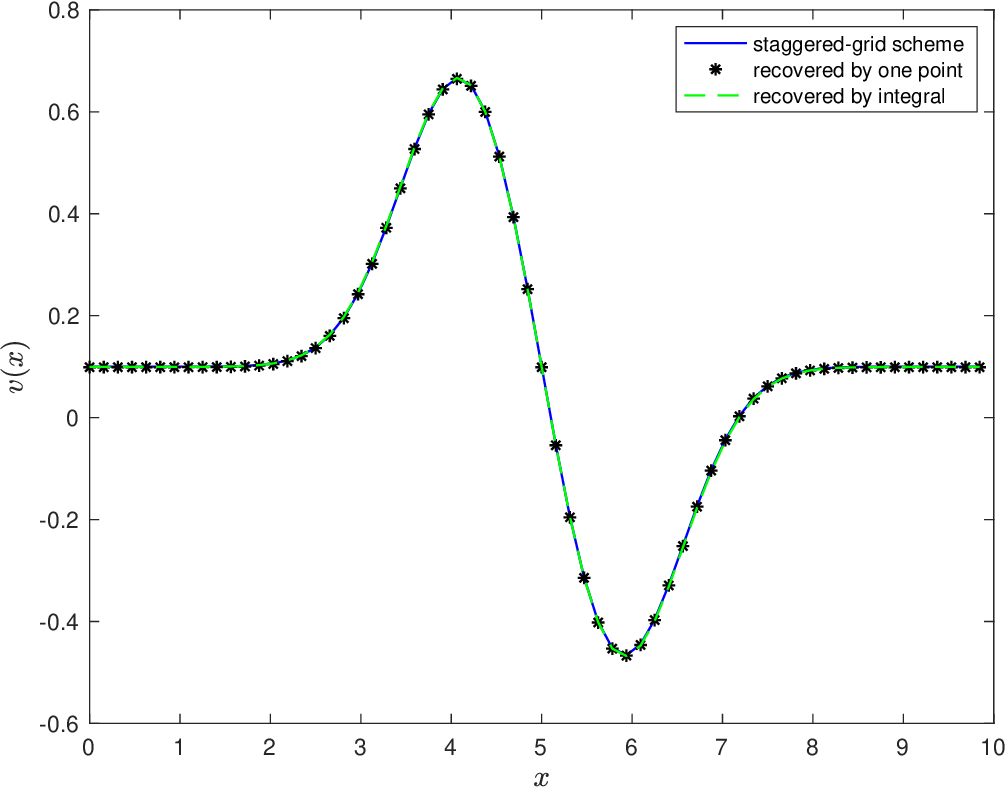}}
    \subfigure[The stress field]{
        \includegraphics[width=0.4\linewidth]{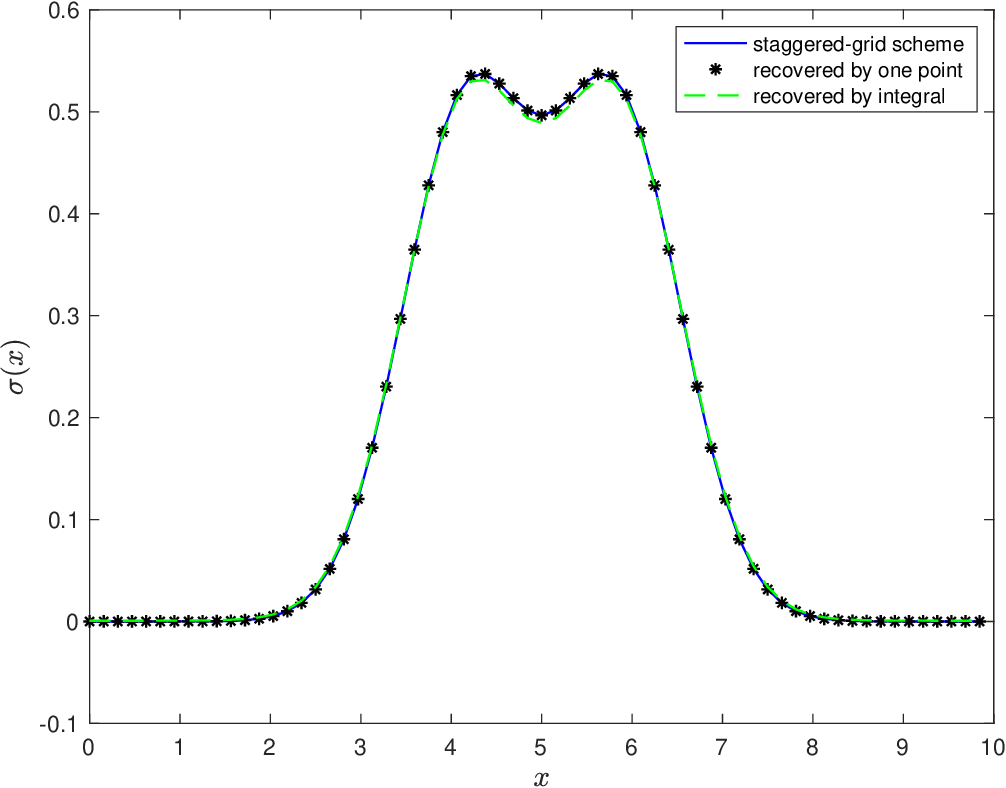}}
    \caption{Velocity and stress fields under the constant external force at $T=1$. }\label{fig:vaco}
\end{figure}

\subsection{Two-dimensional velocity-stress equations with variable coefficients}

We conduct simulations for the 2-dimensional variable coefficient velocity-stress equations. The initial stress $\sigma_{xx}$ is given by the Gaussian function and the medium density and Lam$\acute{e}$ parameters vary with position:
\begin{equation}
    \rho(x) = 1+ 0.5\sin{x}\cos{y},\quad\lambda(x) = 0.5+ 0.2\sin{x}\cos{y},\quad\mu(x) = 0.5+ 0.15\sin{x}\cos{y}.
\end{equation}
In the simulation, an implicit Euler scheme is employed for the temporal discretization, with step size $\Delta t = \frac{1}{200}$, and the spatial discretization numbers are set as $N_x = N_y = 2^5$ by using the staggered grid method. Within the discrete framework, the expanded new variable $ p $ is discretized with $ N_p = 2^{10} $ grid  points over the interval $[-3\pi, 3\pi]$. Figures \ref{fig:DFT_vs2} presents the comparison between the classical results and the quantum results at $x^* = 5$ and $y^* = 5$, which shows the Schr\"odingerisation with staggered grid method can match well.

\begin{figure}[htbp]
    \centering
    \subfigure[$v_1(x,y)$.]{
        \includegraphics[width=0.3\linewidth]{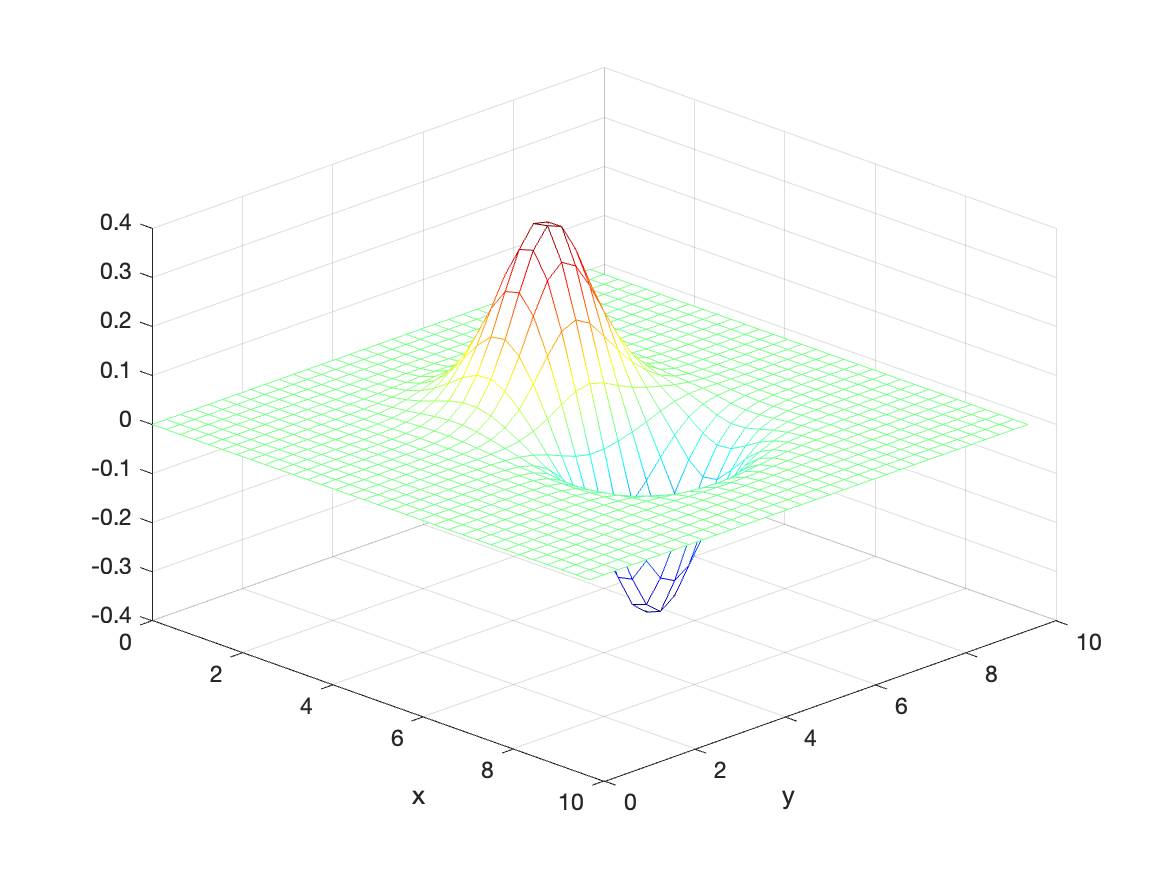}\label{fig:v1_2D}}
    \subfigure[$v_1(x^*,y)$.]{
        \includegraphics[width=0.27\linewidth]{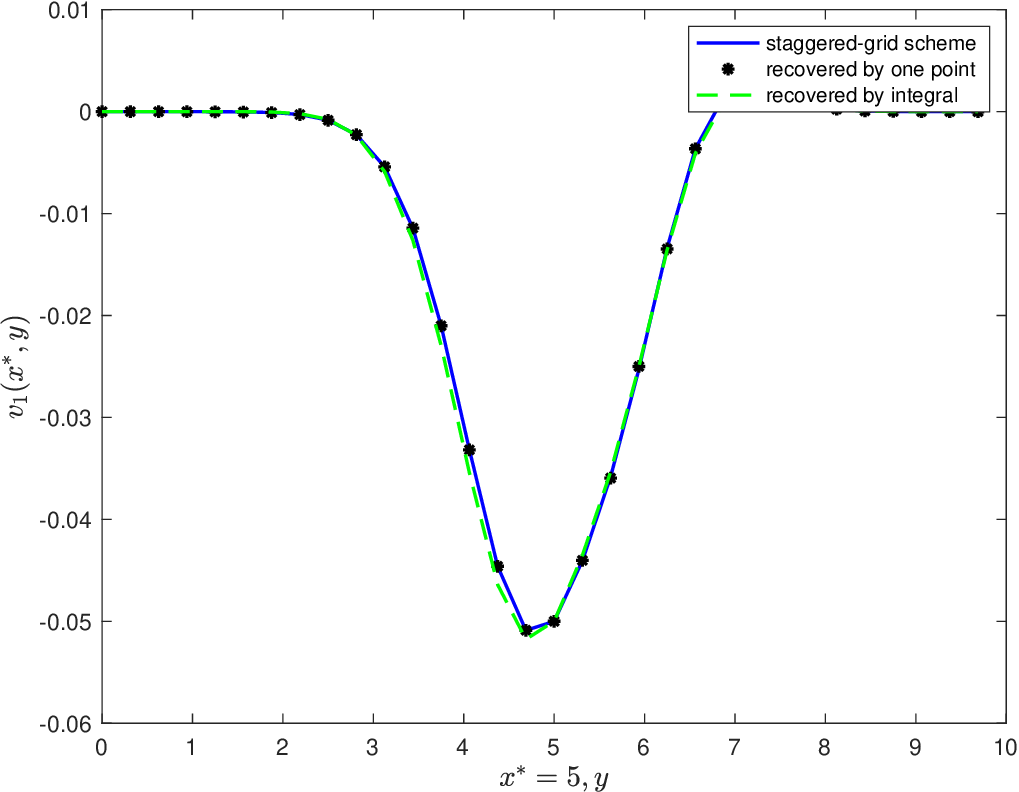}\label{fig:v1_fixX}}
    \subfigure[$v_1(x,y^*)$]{
        \includegraphics[width=0.27\linewidth]{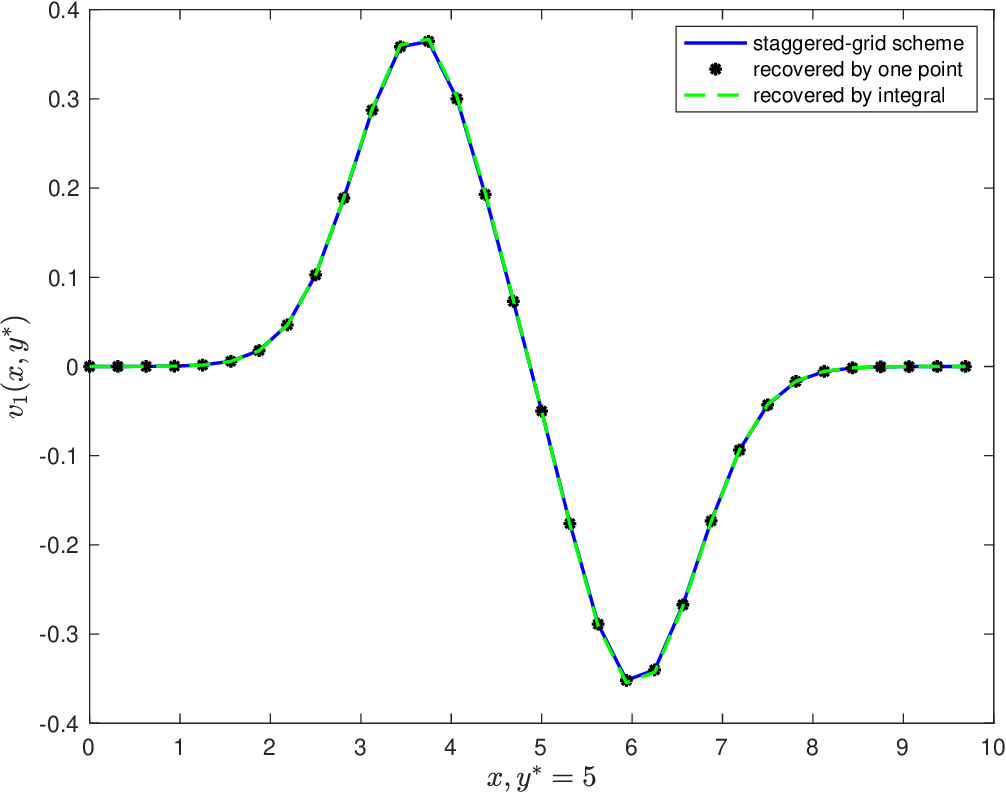}\label{fig:v1_fixY}}
    \subfigure[$v_2(x,y)$.]{
    \includegraphics[width=0.3\linewidth]{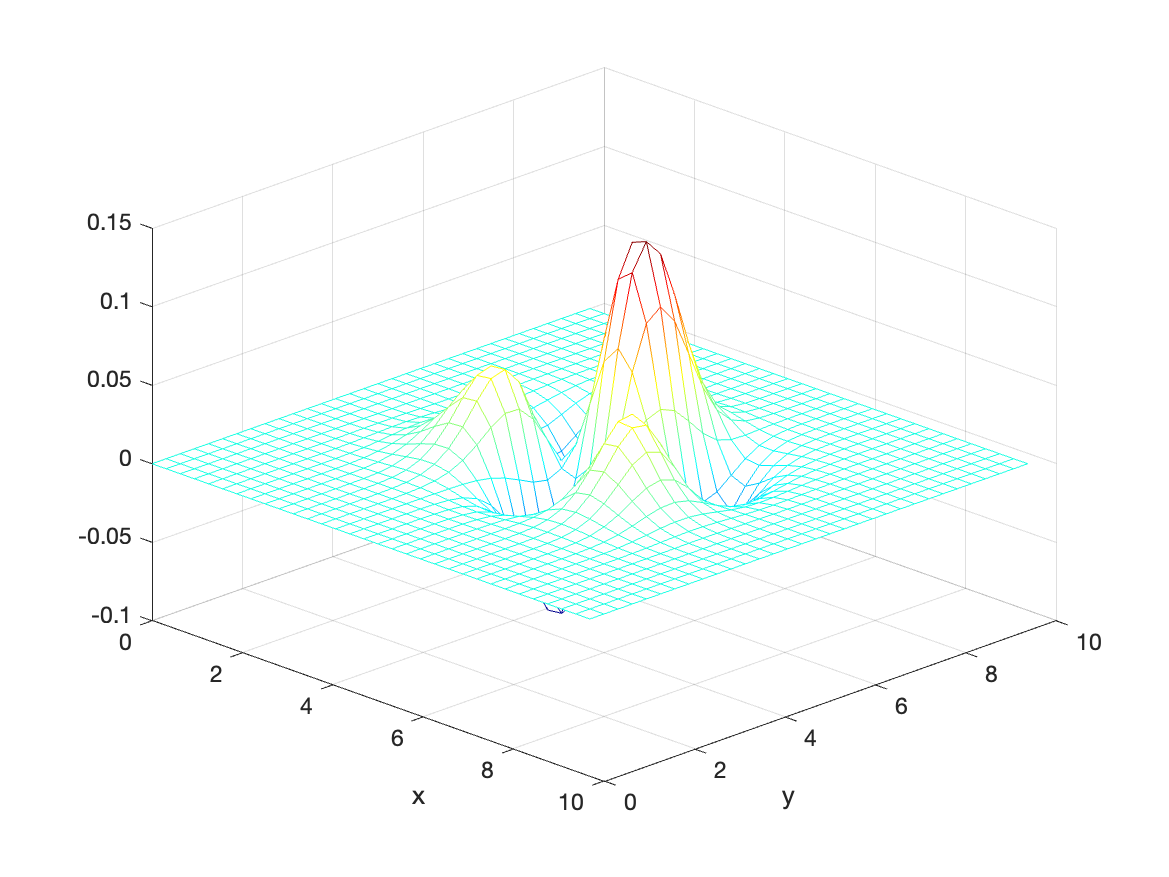}\label{fig:v2_2D}}
    \subfigure[$v_2(x^*,y)$.]{
    \includegraphics[width=0.27\linewidth]{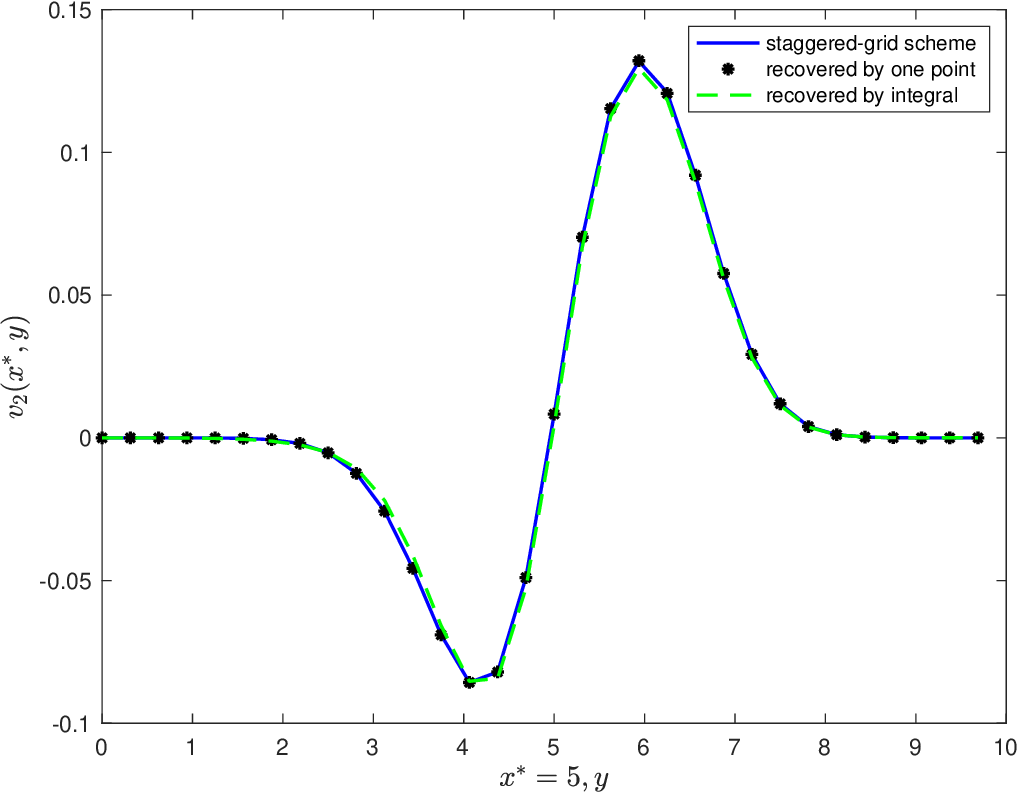}\label{fig:v2_fixX}}
    \subfigure[$v_2(x,y^*)$]{
    \includegraphics[width=0.27\linewidth]{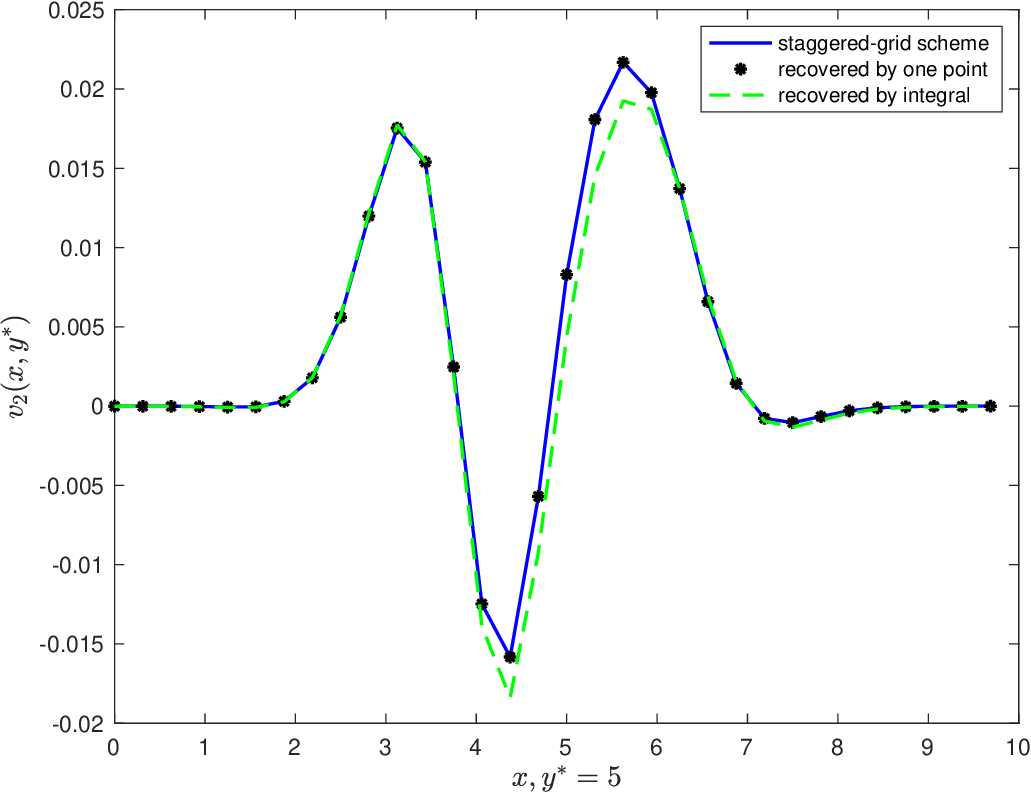}\label{fig:v2_fixY}}

    \subfigure[$\sigma_1(x,y)$.]{
    \includegraphics[width=0.3\linewidth]{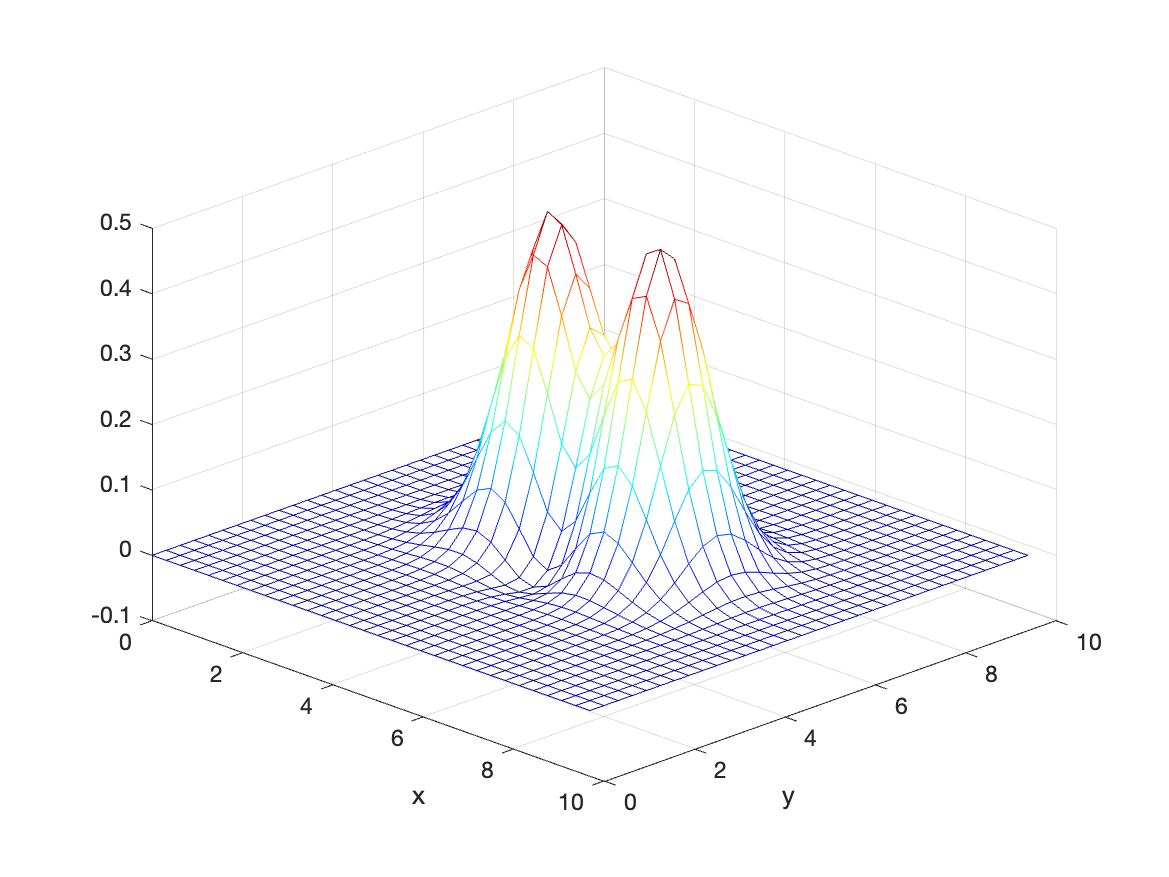}\label{fig:c1_2D}}
    \subfigure[$\sigma_1(x^*,y)$.]{
    \includegraphics[width=0.27\linewidth]{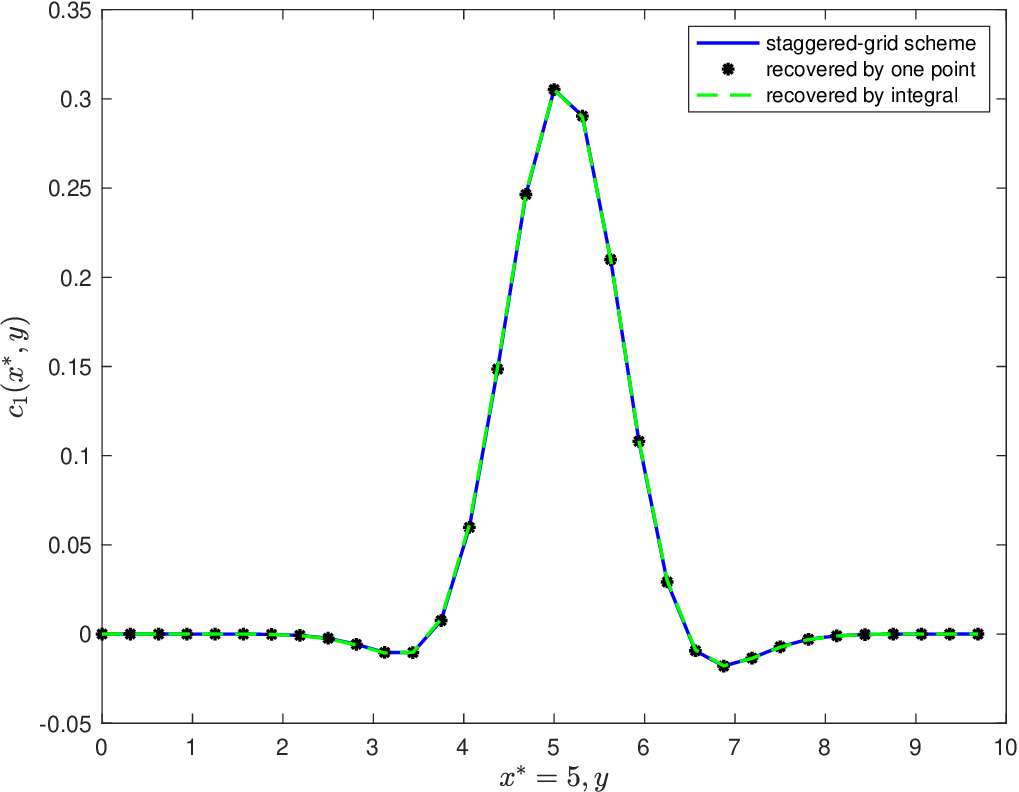}\label{fig:c1_fixX}}
    \subfigure[$\sigma_1(x,y^*)$]{
    \includegraphics[width=0.27\linewidth]{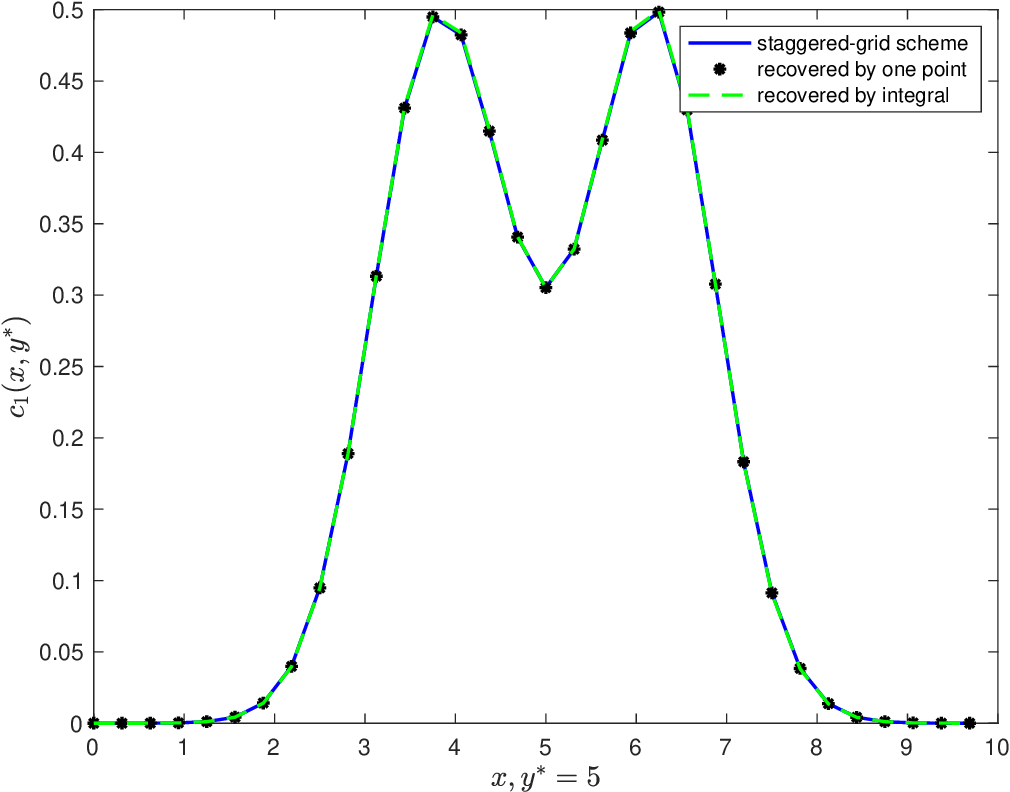}\label{fig:c1_fixY}}

    \subfigure[$\sigma_2(x,y)$.]{
    \includegraphics[width=0.3\linewidth]{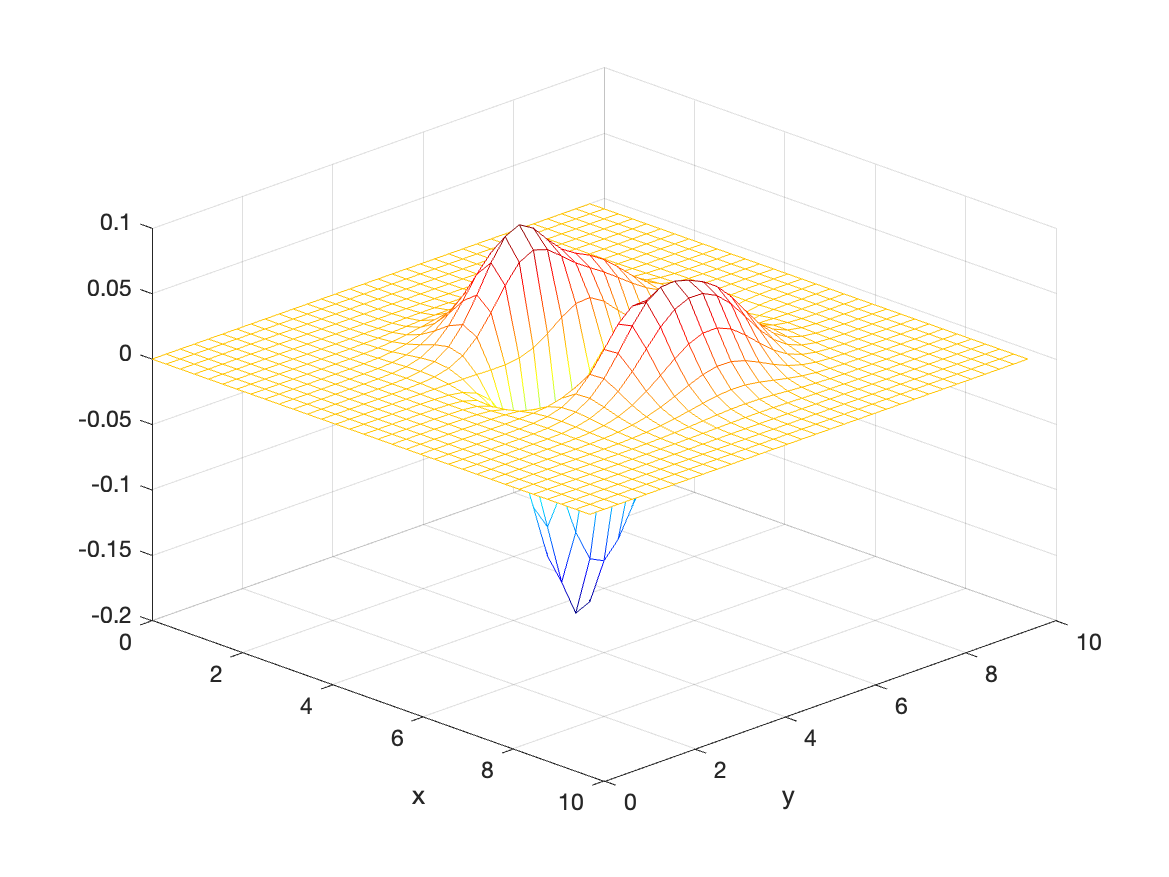}\label{fig:c2_2D}}
    \subfigure[$\sigma_2(x^*,y)$.]{
    \includegraphics[width=0.27\linewidth]{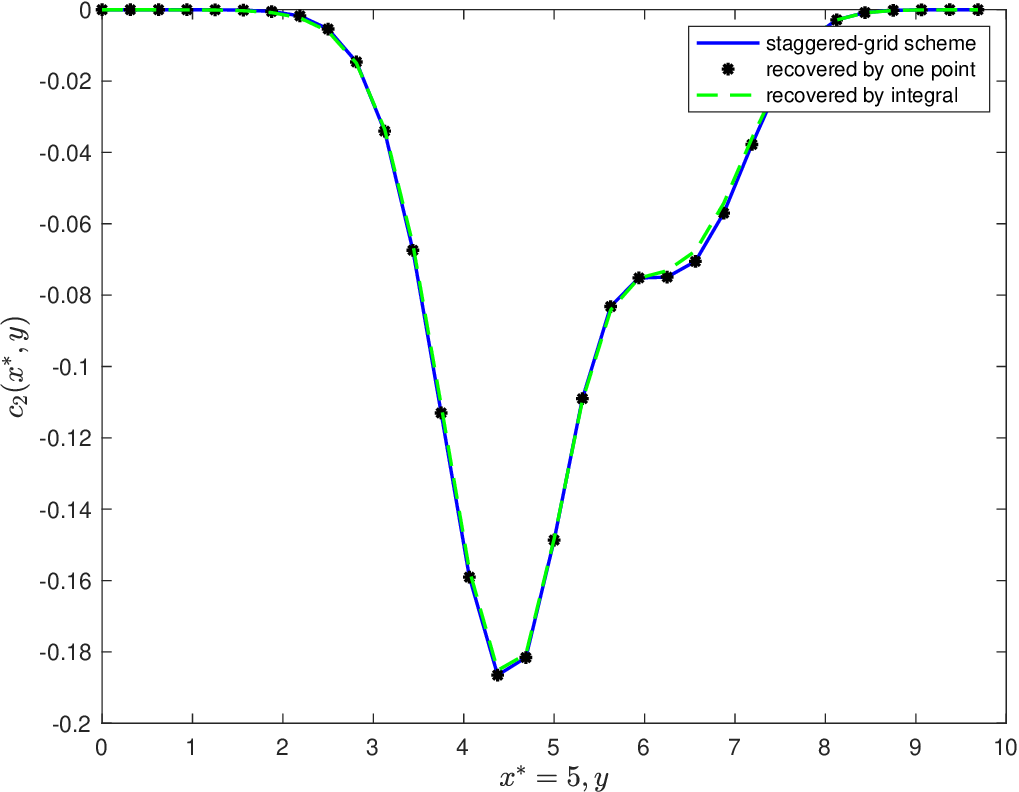}\label{fig:c2_fixX}}
    \subfigure[$\sigma_2(x,y^*)$]{
    \includegraphics[width=0.27\linewidth]{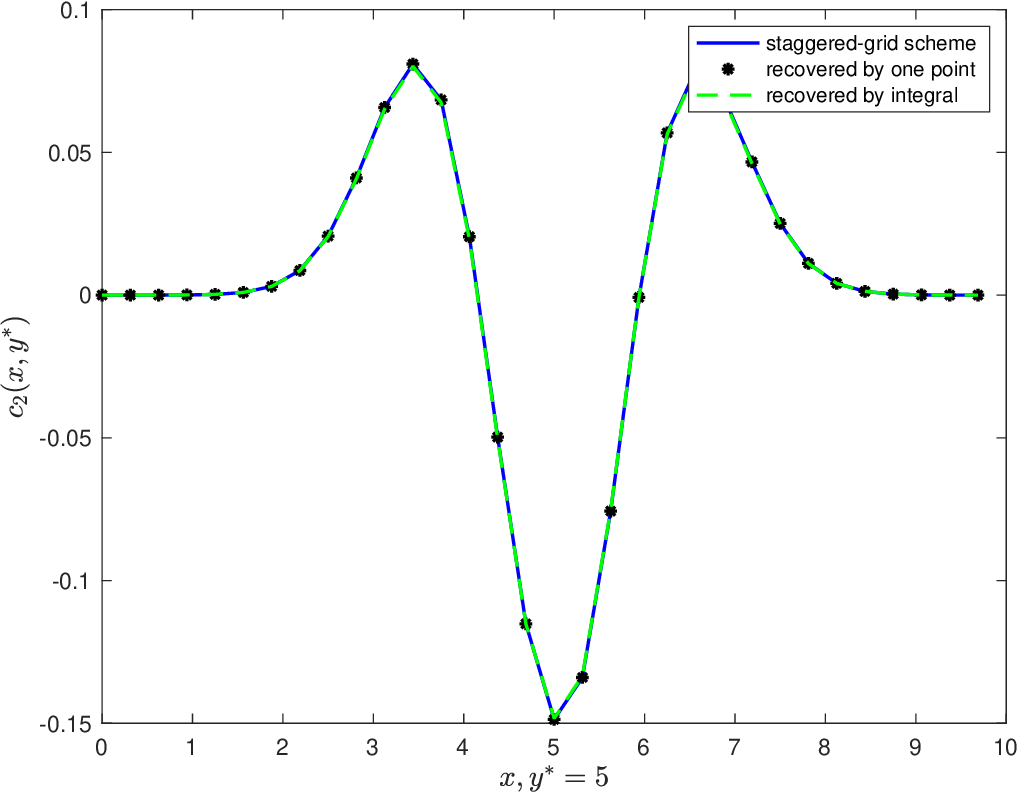}\label{fig:c2_fixY}}

    \subfigure[$\sigma_3(x,y)$.]{
    \includegraphics[width=0.3\linewidth]{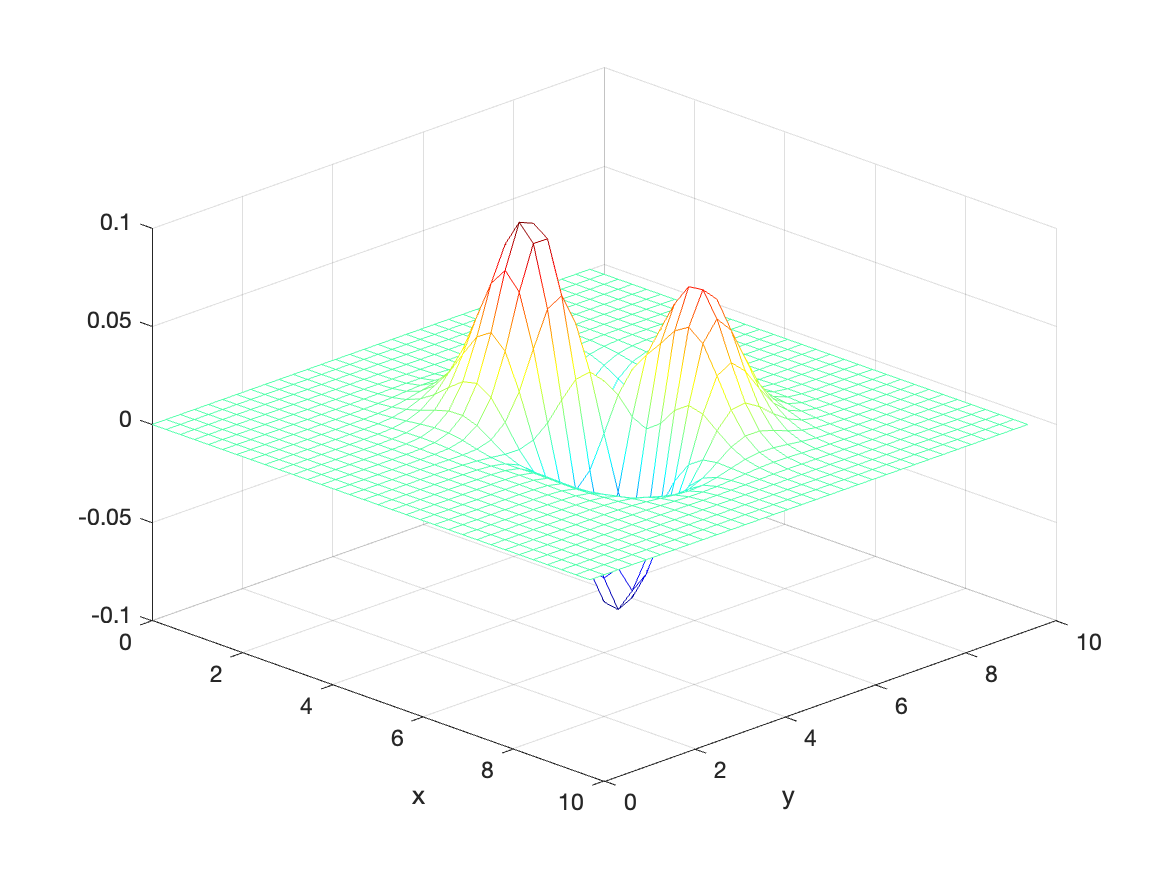}\label{fig:c3_2D}}
    \subfigure[$\sigma_3(x^*,y)$.]{
    \includegraphics[width=0.27\linewidth]{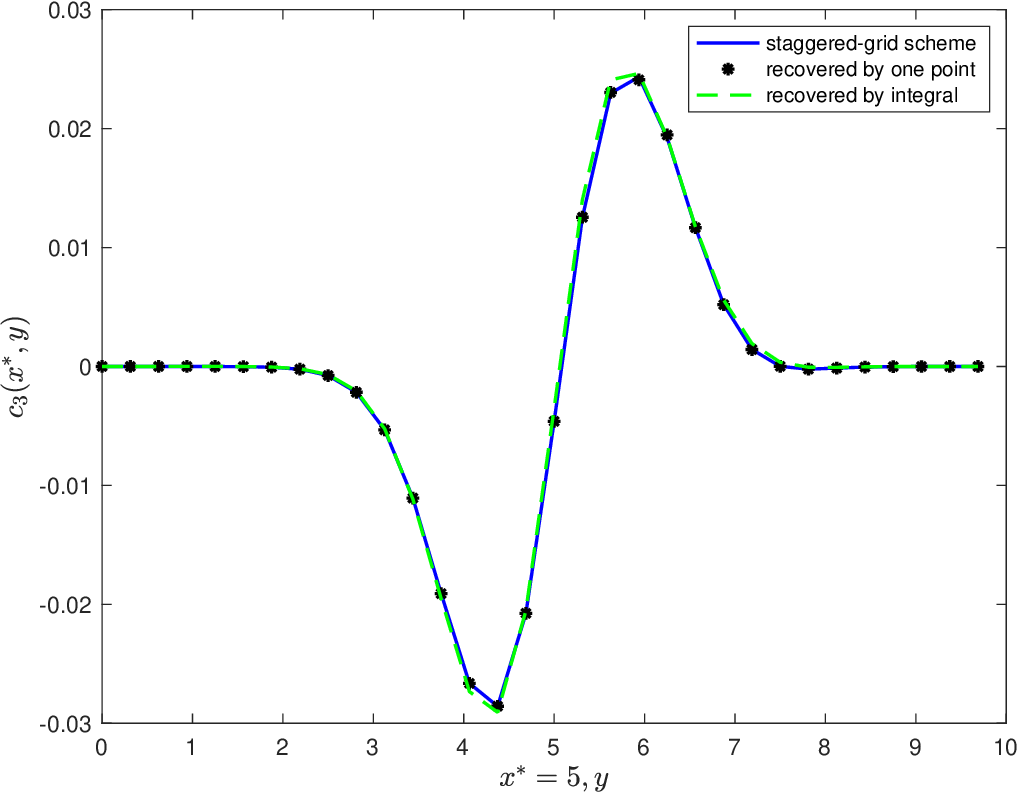}\label{fig:c3_fixX}}
    \subfigure[$\sigma_3(x,y^*)$]{
    \includegraphics[width=0.27\linewidth]{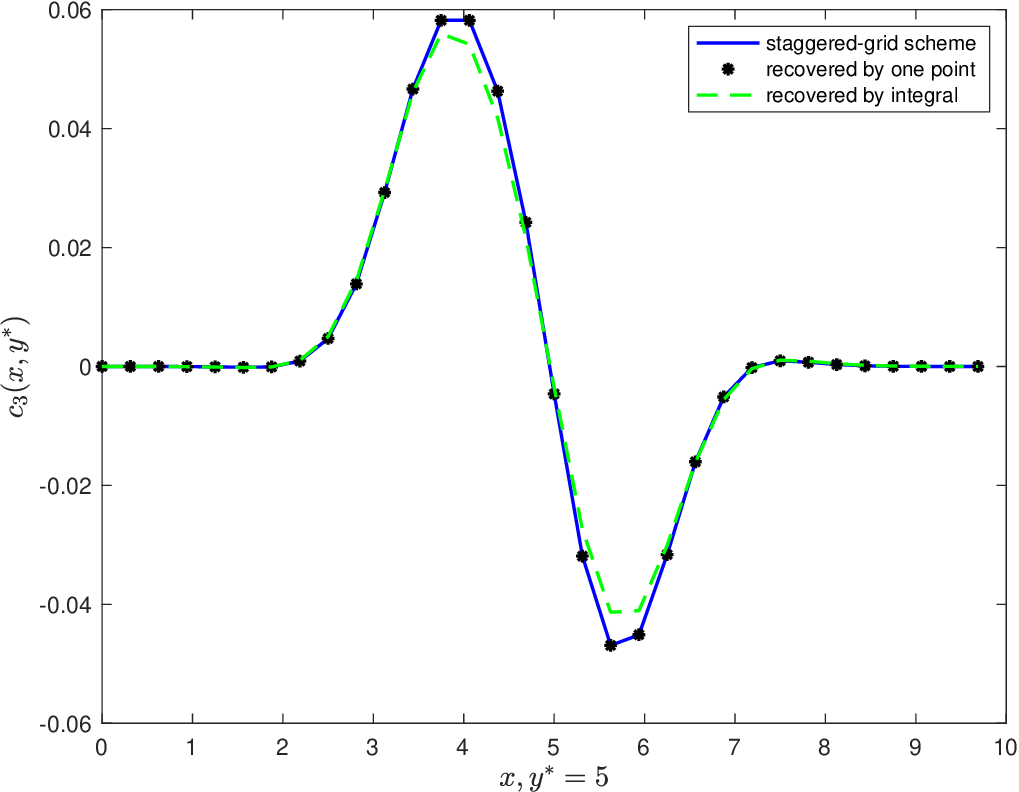}\label{fig:c3_fixY}}

    \caption{Result at $T=1$ with periodic boundary condition. On the left are the numerical solutions obtained by staggered grid method, in the middle and right are the cross-section solutions at \( x^*=5 \) and $y^*=5$. The top two rows present the velocity field, and the bottom three rows present the stress field.}\label{fig:DFT_vs2}

\end{figure}

\subsection{One-dimensional wave equation with periodic boundary conditions}

In $1$-dimensional case, the elastic wave propagation in a homogeneous medium is governed by the classical wave equation. Considering its hyperbolic system representation:

\begin{equation}\label{numtest_1}
    \begin{cases}
    \rho u_{tt} = \frac{\partial\sigma_{11}}{\partial x},\quad 0<x<1, t>0, \notag\\
    \sigma_{11} = \lambda u_x+2\mu u_x.
    \end{cases}
\end{equation}
The exact solution is $u = \cos{2 \pi t}\sin{2 \pi x}+(\sin{4 \pi t}\sin{4\pi x})/4 \pi$. We use Eq. \eqref{w_spectral_Schr} and Eq. \eqref{w_central_Schr} to perform numerical tests on the hyperbolic system with the Crank-Nicolson scheme and the step size is set as $\Delta t =\frac{1}{2000}$. Moreover, we adopt two parameter choices as in Tab. \ref{Tab:parameter} that affect the choice of  $p^*$  to recover the solution.
\begin{table}
\centering
\begin{tabular}{l|ccc}
Item & $\rho$ & $\mu$ & $\lambda$ \\\hline
1 & 1.41&0.35&0.71\\
2 & 1.41&0.40&0.61\end{tabular}
\caption{\label{Tab:parameter}Parameters of medium}
\end{table}
Eqs. \eqref{numtest_1} are equivalent to the following system:
\begin{equation}\label{eq46}
    \frac{d}{dt}\begin{pmatrix}
        \xi\\
        \epsilon\\
        p
    \end{pmatrix} = \begin{pmatrix}
         0&1 / \rho &1 / \rho\\
         2\mu &0&0\\
         \lambda &0 &0 
    \end{pmatrix}\begin{pmatrix}
        \xi\\
        \epsilon\\
        p
    \end{pmatrix}_x,
\end{equation}
where $\rho=\lambda+2\mu$. The exact solutions and the initial conditions can also be derived, in addition to the  period boundary conditions.
$$
\begin{aligned}
    \xi(x,t)&=-4\pi\sin{4\pi t}\sin{4\pi x}+\cos{8\pi t}\sin{8 \pi x}, \\
    \epsilon(x,t) &= 2\mu (4\pi\cos{4\pi t}\cos{4\pi x}+\sin{8\pi t}\cos{8\pi x}), \\
    p(x,t) &= \lambda\epsilon(x,t)/(2\mu).
\end{aligned}
$$

\subsubsection{Quantum simulation with spectral method}

In this part, we use Schr$\ddot{\text{o}}$dingerisation combined with spectral method to simulate the hyperbolic system. We set the spatial grid number $N_x = 2^5$ and the extended variable domain is $p\in[-3\pi,3\pi]$ and $N_p = 2^9 $. When choosing parameters in the  first row of Tab. \ref{Tab:parameter}, the Hermite matrix $H_1$ is non-positive definite matrix, so we pick $p_1=0.037\geq p^* = 0$ in Schr$\ddot{\text{o}}$dingerisation method. When choosing the parameters in the second row of Tab. \ref{Tab:parameter}, the Hermite matrix $H_1$ has positive eigenvalues, so we pick $p_1=6.774\geq p^* = 6.759$. The results are shown in Figures \ref{fig1}. Each subfigure shows the results of the classical scheme, the exact solution, and two recovery methods. They are in good agreements. 

\begin{figure}[htbp]
    \centering
    \subfigure[$\xi(x, T)$]{
        \includegraphics[width=0.3\linewidth]{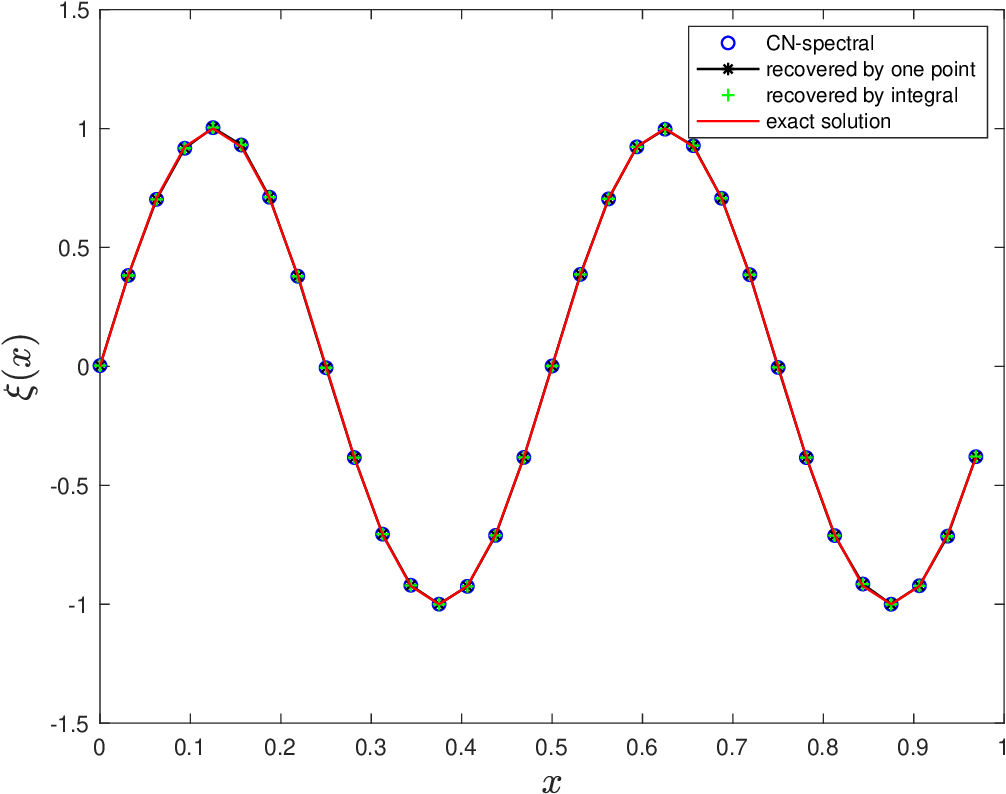}}
    \subfigure[$\epsilon(x, T)$]{
        \includegraphics[width=0.3\linewidth]{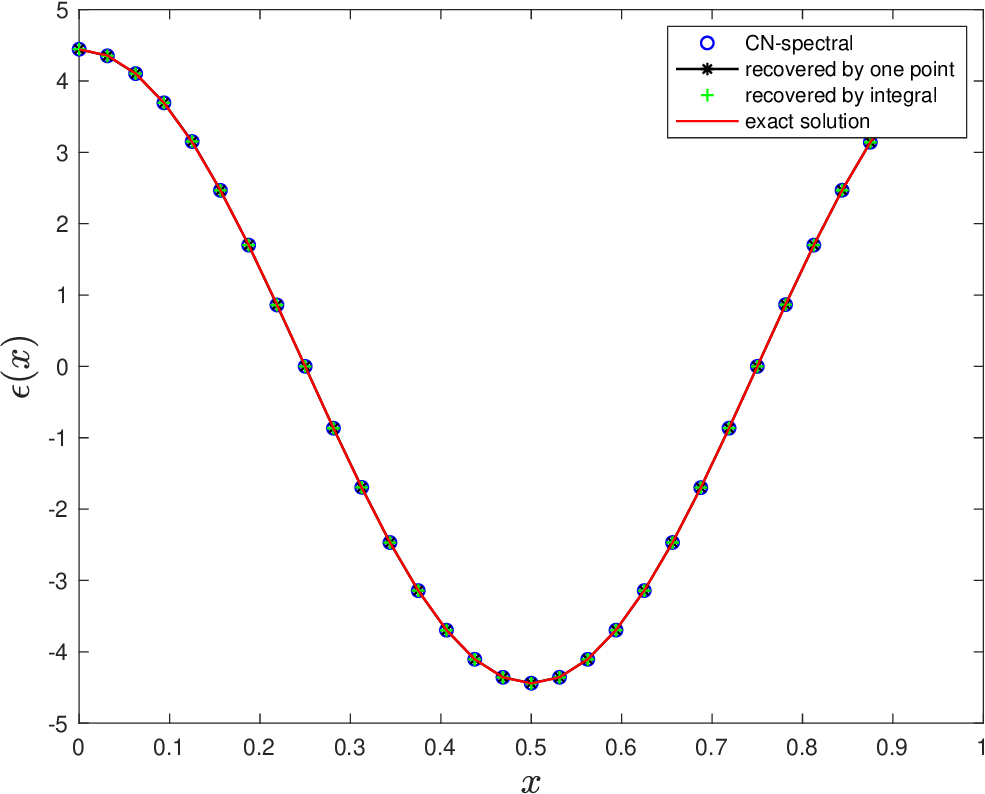}}
    \subfigure[$p(x, T)$]{
        \includegraphics[width=0.3\linewidth]{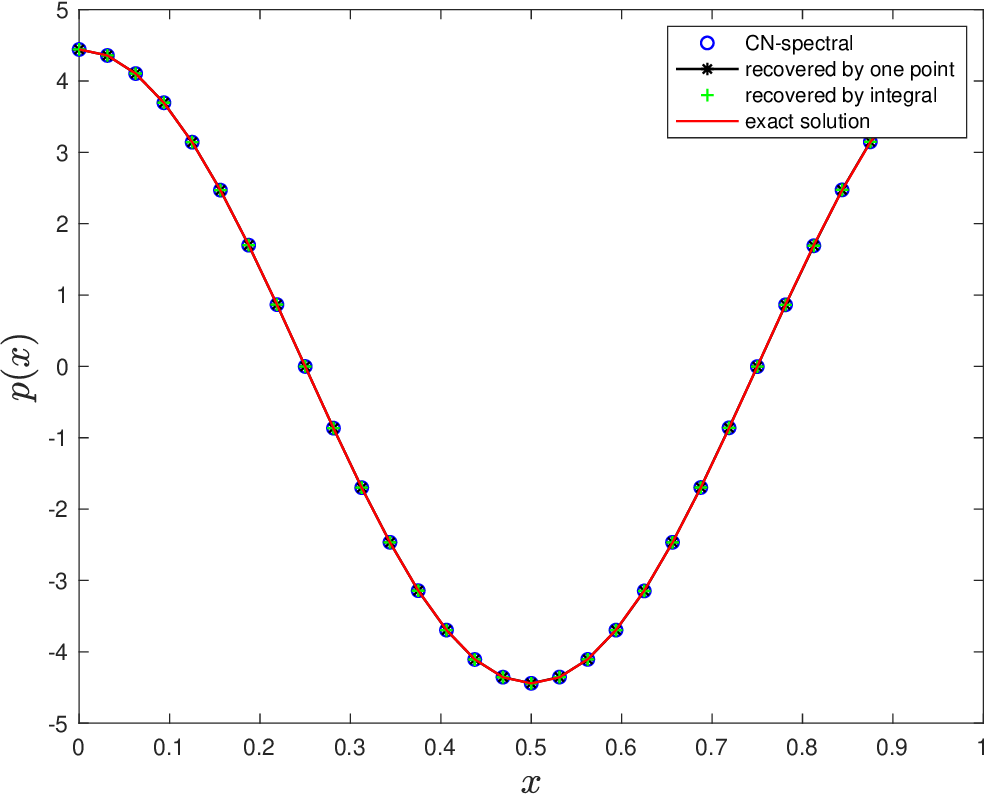}}
    \subfigure[$\xi(x, T)$]{
        \includegraphics[width=0.3\linewidth]{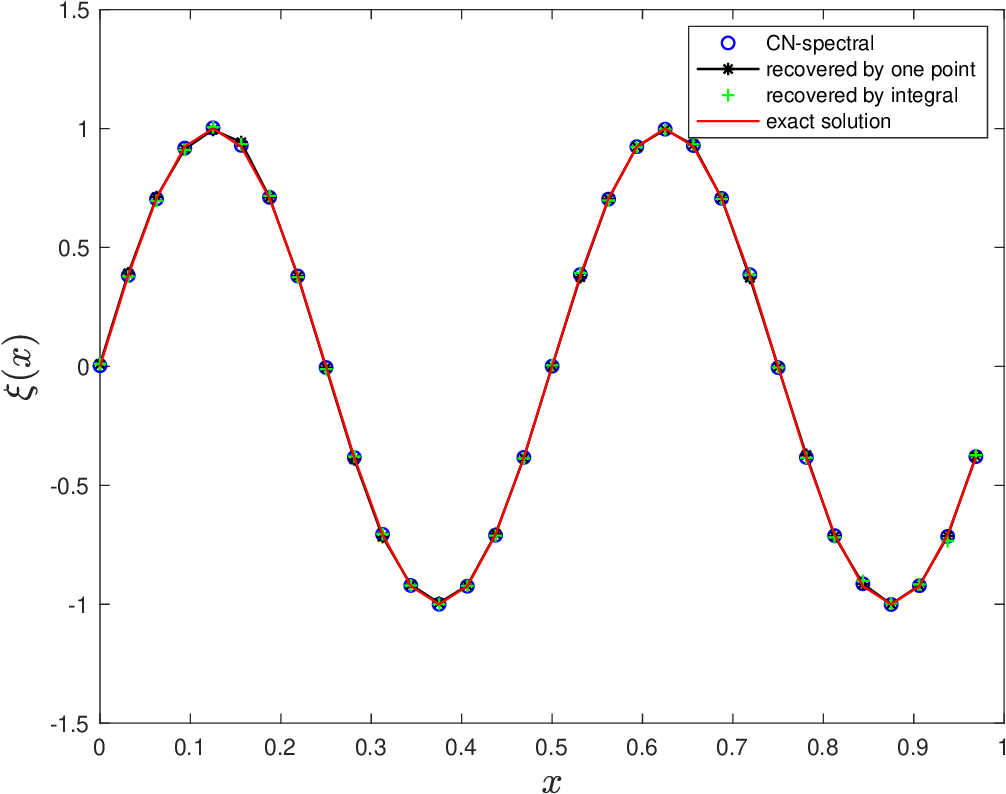}}
    \subfigure[$\epsilon(x, T)$]{
        \includegraphics[width=0.3\linewidth]{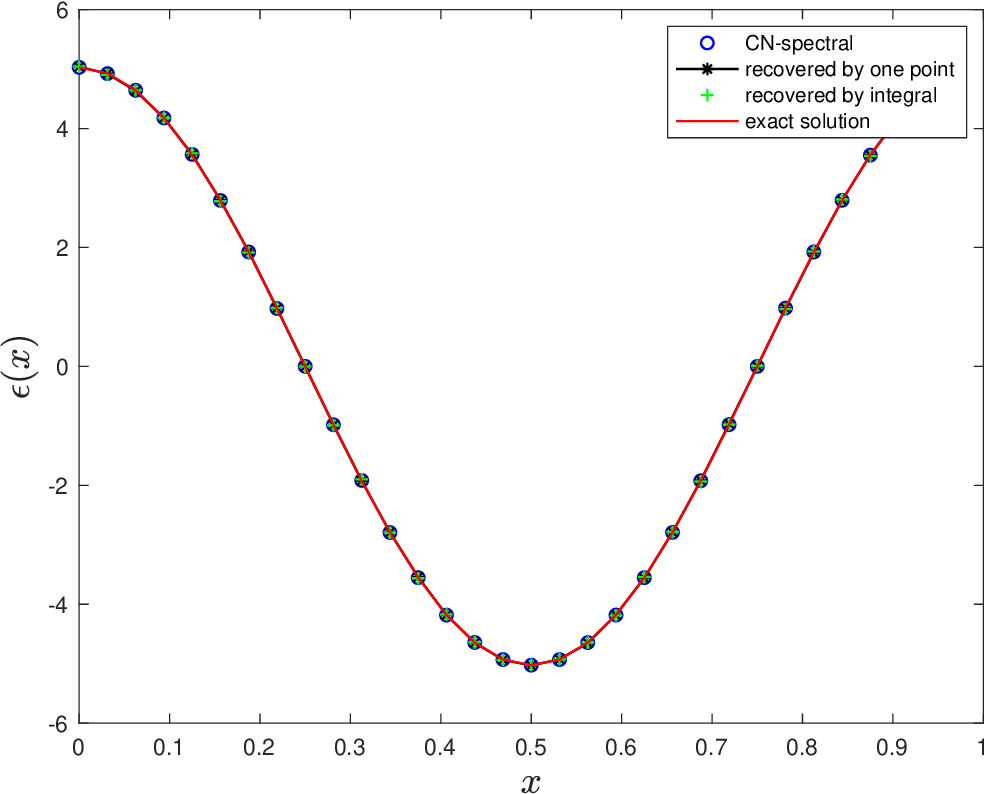}}
    \subfigure[$p(x, T)$]{
        \includegraphics[width=0.3\linewidth]{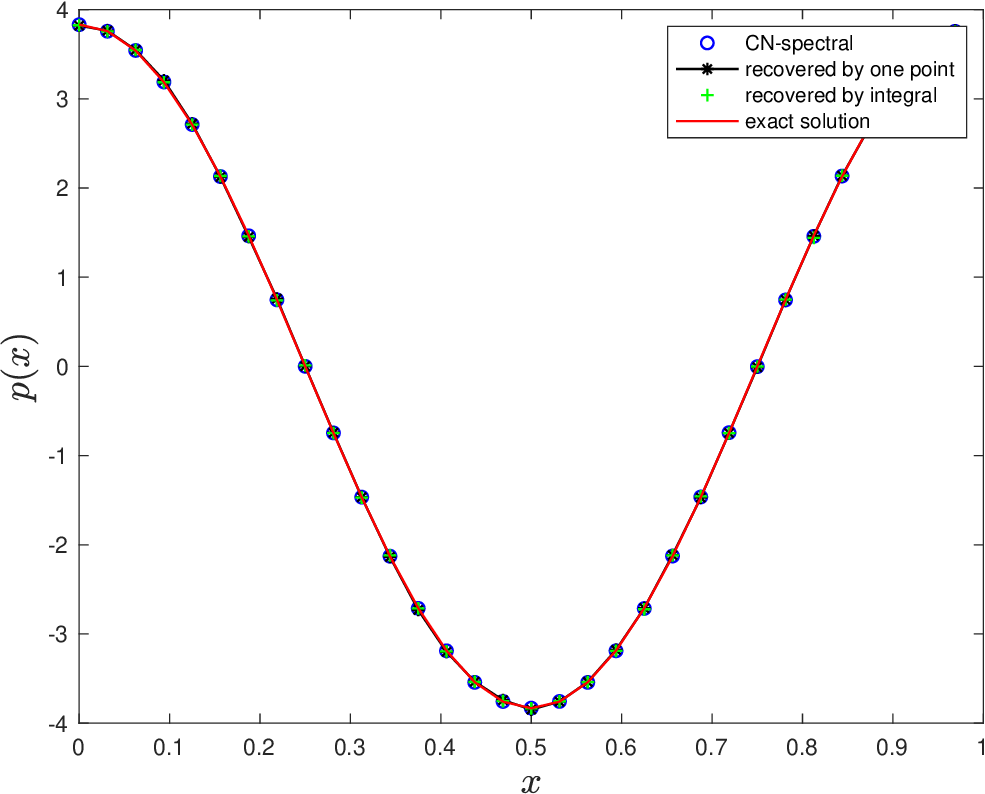}}
    \caption{Results of the hyperbolic system at $T=1$ using the Schr$\ddot{\text{o}}$dingerisation with the spectral method. The first row: using the first row parameters in Tab. \ref{Tab:parameter} ; The second row: using the second parameters in Tab. \ref{Tab:parameter}.}\label{fig1}
\end{figure}

\subsubsection{Quantum simulation with central difference method}

The central difference method has lower accuracy than the spectral method in space, so we set $N_x=2^6$ for numerical tests. In this part, we use the same setup as the above experiment. The spatial scheme will affect the eigenvalues of $H_1$, so for the second row parameters, the maximum eigenvalue of $H_1$ is $4.303$, so we pick $p_1=4.306$ in Schr$\ddot{\text{o}}$dingerisation method. Figures \ref{fig3} depict a comparison plot of the simulation results with the classical algorithm, quantum algorithm and the exact solution.
\begin{figure}[htbp]
    \centering
    \subfigure[$\xi(x, T)$]{
        \includegraphics[width=0.3\linewidth]{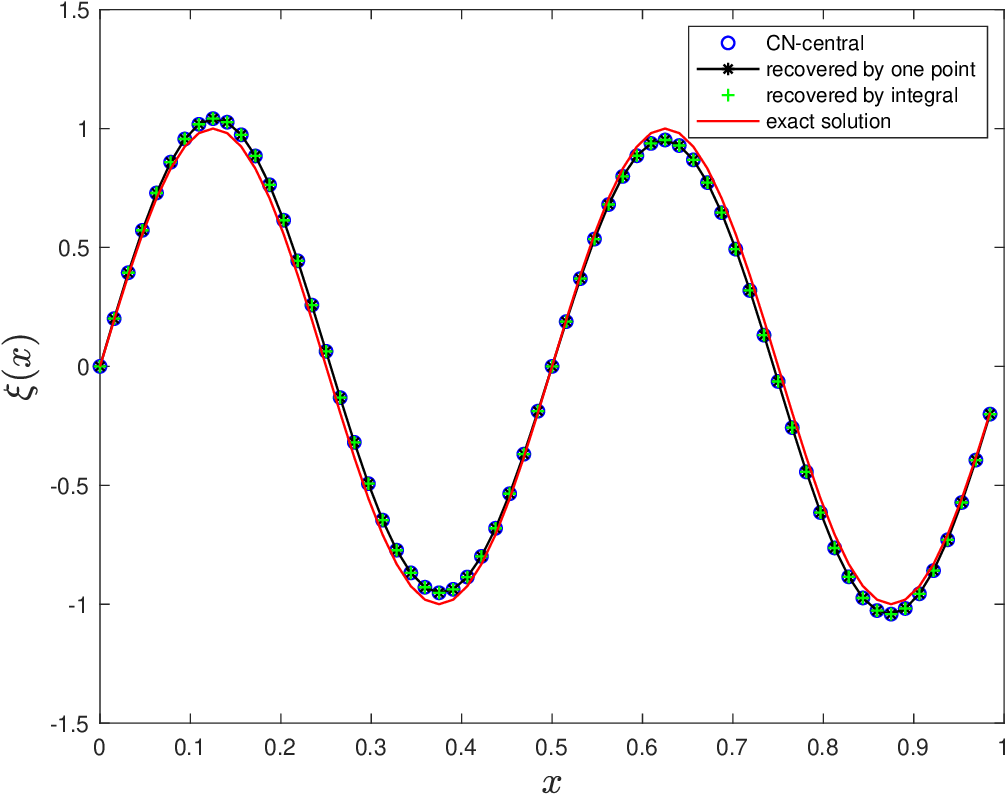}}
    \subfigure[$\epsilon(x, T)$]{
        \includegraphics[width=0.3\linewidth]{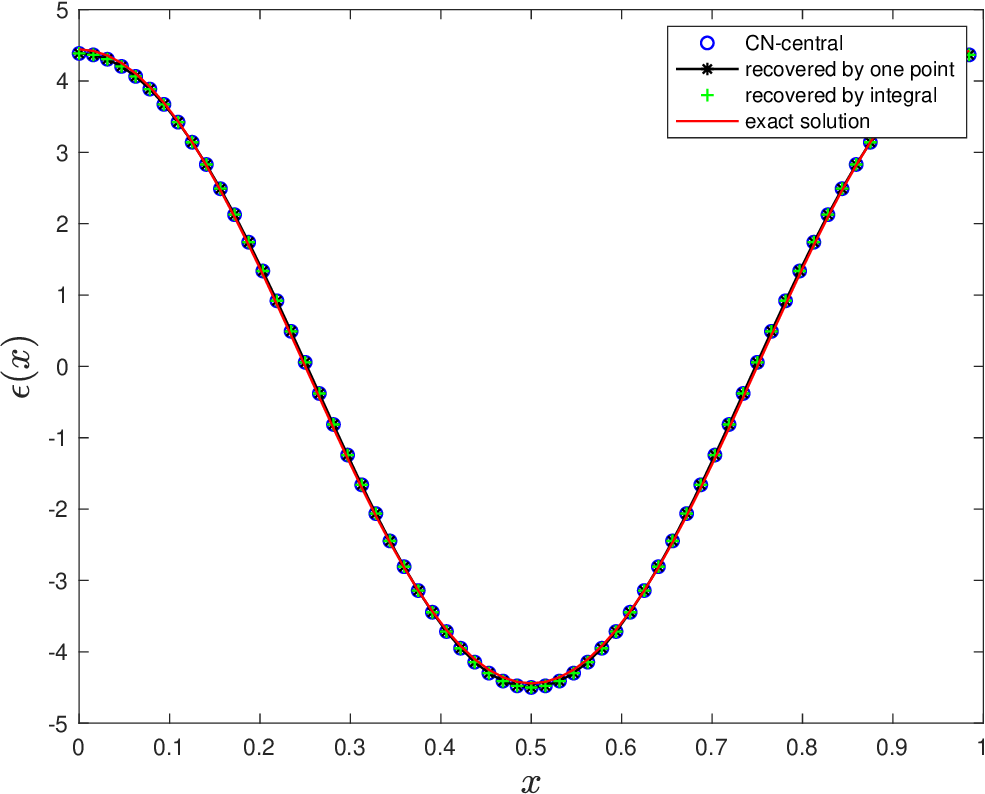}}
    \subfigure[$p(x, T)$]{
        \includegraphics[width=0.3\linewidth]{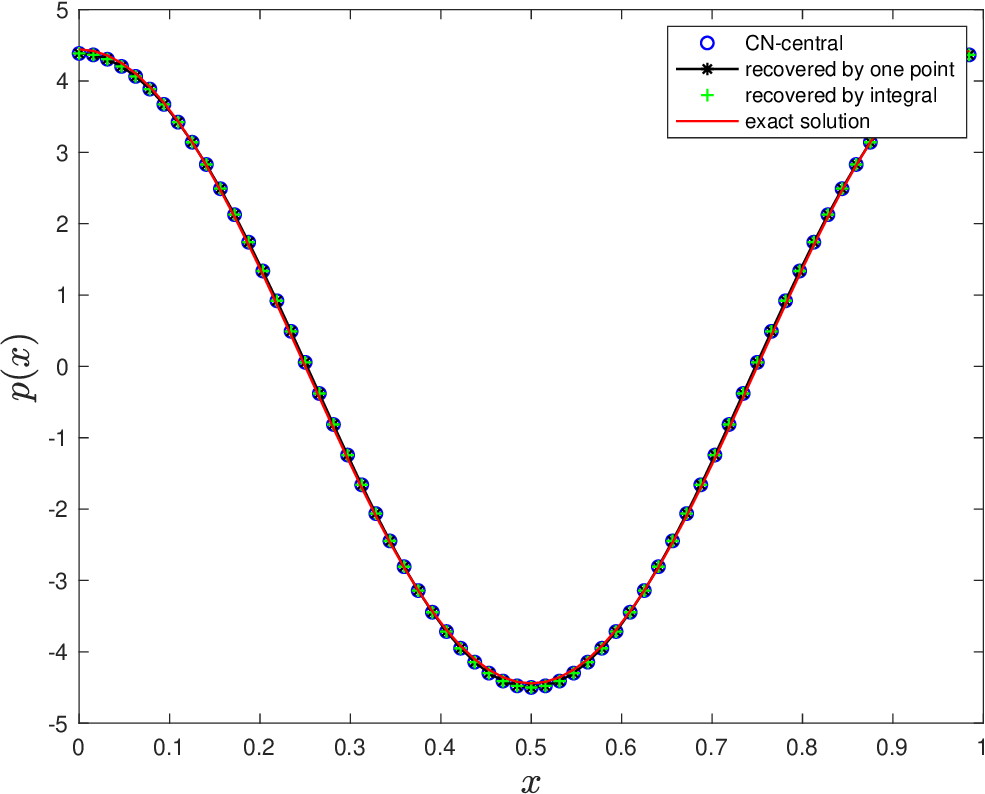}}
    \subfigure[$\xi(x, T)$]{
        \includegraphics[width=0.3\linewidth]{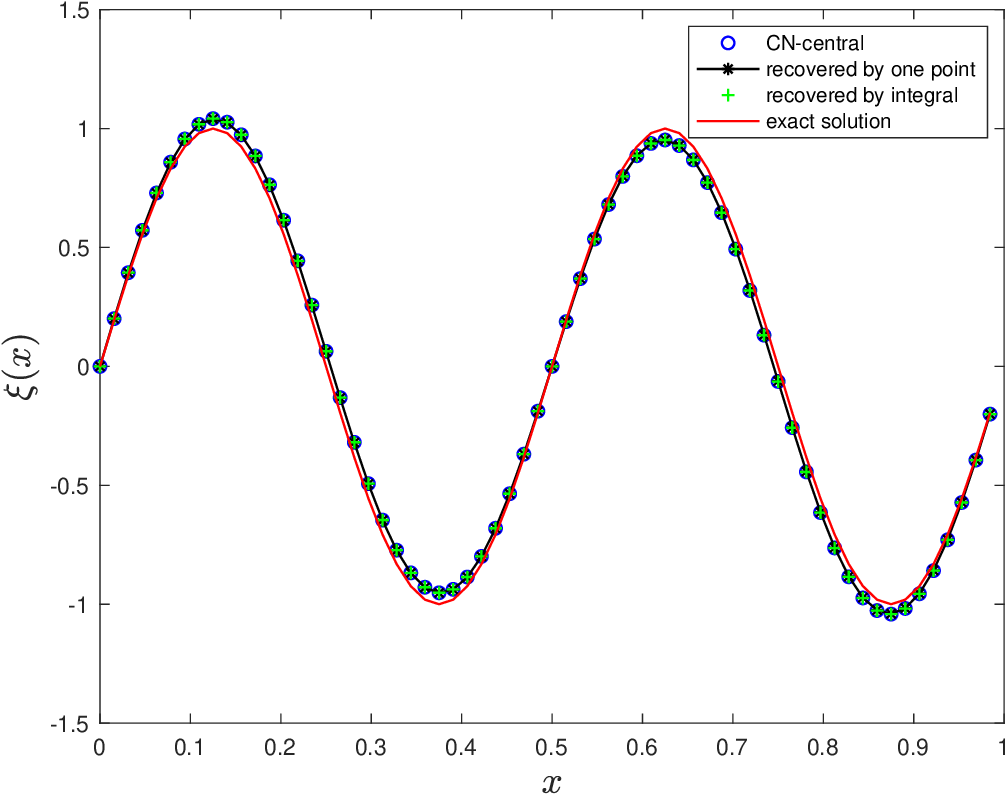}}
    \subfigure[$\epsilon(x, T)$]{
        \includegraphics[width=0.3\linewidth]{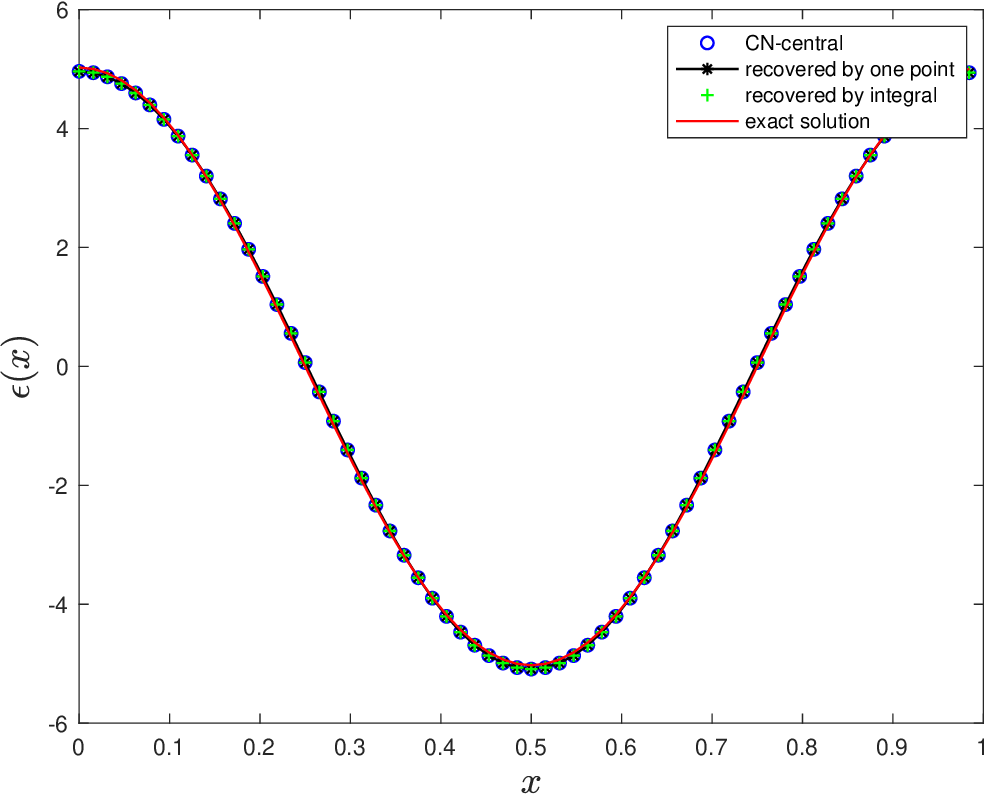}}
    \subfigure[$p(x, T)$]{
        \includegraphics[width=0.3\linewidth]{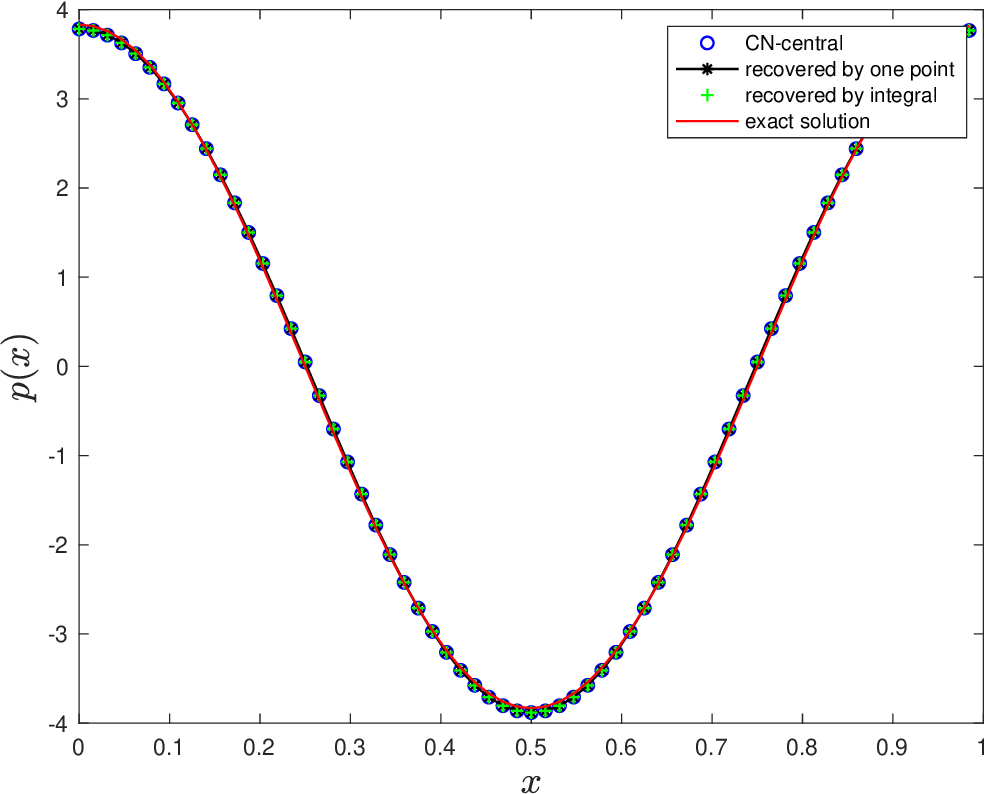}}
    \caption{Results of the hyperbolic system at $T=1$ using the Schr$\ddot{\text{o}}$dingerisation with the central difference method. The first row: using the first row parameters in Tab. \ref{Tab:parameter} ; The second row: using the second parameters in Tab. \ref{Tab:parameter}.}\label{fig3}
\end{figure}

\section{Conclusion}
In this paper, we develop a quantum simulation framework for elastic wave equations using the Schr$\ddot{\text{o}}$dingerisation method, addressing both the velocity-stress equations (expressed as symmetric hyperbolic systems) and the wave displacement equations (transformed into higher-dimensional hyperbolic systems). By mapping these classical PDEs to quantum Hamiltonian systems, we analyze the gate complexity by Schr$\ddot{\text{o}}$dingerisation method. Our results establish a theoretical foundation for quantum-accelerated modeling of elastic wave propagation, with potential applications in geophysical exploration and materials science.

Realistic wave propagation often involves complex boundary value problems and unbounded domains. Our future efforts will focus on quantum algorithms for such generalized scenarios and derive the corresponding quantum circuit for practical simulations.

\newpage

\bibliographystyle{alpha}
\bibliography{sample}

\end{document}